\documentclass[11pt]{article}

\usepackage{amsthm,amsfonts,amsmath,times}
\usepackage{algonjk,multicol}
\usepackage{enumerate}
\usepackage{hyperref}
\def\Comment#1{\textsl{$\langle\!\langle$#1\/$\rangle\!\rangle$}}

\advance \topmargin by -1.0in
\advance \oddsidemargin by -0.8in
\newcommand{\myparskip}{3pt}
\parskip \myparskip

\newtheorem{lemma}{Lemma}[section]
\newtheorem{theorem}[lemma]{Theorem}

\newtheorem{definition}[lemma]{Definition}
\newtheorem{corollary}[lemma]{Corollary}
\newtheorem{proposition}[lemma]{Proposition}

\newtheorem{remark}[lemma]{Remark}

\newcommand{\opt}{\textsc{OPT}}
\newcommand{\etal}{{\em et al.}\ }

\def\eps{\varepsilon}

\def\script#1{\mathcal{#1}}
\def\C{\script{C}}
\def\S{\script{S}}
\def\T{\script{T}}
\def\D{\script{D}}
\def\pot{\textrm{Potential}}
\def\unc{\textrm{Uncharged}}
\def\hist{\textrm{History}}
\def\rem{\textrm{RemainingCharge}}
\def\len{\textsc{Length}}
\def\pen{\textsc{Penalty}}
\def\cost{\textsc{Cost}}

\def\Comment#1{\textsl{$\langle\!\langle$#1\/$\rangle\!\rangle$}}

\renewenvironment{proof}{\vspace{-0.1in}\noindent{\bf Proof:}}%
        {\hspace*{\fill}$\Box$\par}
\newenvironment{proofof}[1]{\smallskip\noindent{\bf Proof of #1:}}%
        {\hspace*{\fill}$\Box$\par}
\newenvironment{proofsketch}{\vspace{-0.05in}\noindent{\bf Proof Sketch:}}%
        {\hspace*{\fill}$\Box$\par}

\begin{document}

\title{Prize-Collecting Steiner Tree and Forest in Planar
  Graphs\footnote{In independent work, Bateni, Hajiaghayi and Marx
    \cite{BateniHM-pcst} have obtained several results on the topics
    considered in this paper.  More details can be found in related
    work in Section~\ref{sec:intro}.}  }
\author{
Chandra Chekuri\thanks{Dept. of Computer Science, University of Illinois, 
Urbana, IL 61801. Supported in part by NSF grants CCF 0728782 and 
CNS-0721899. {\tt chekuri@cs.illinois.edu}}
\and
Alina Ene\thanks{Dept. of Computer Science, University of Illinois, Urbana,
    IL 61801. Supported in part by NSF grant CCF 0728782. 
{\tt  ene1@illinois.edu}}
 \and 
Nitish Korula\thanks{Dept. of Computer Science, University of Illinois, Urbana,
    IL 61801. Supported in part by a dissertation completion fellowship from
 the Univ.\ of Illinois. {\tt nkorula2@illinois.edu}} 
}

\date{\today}

\maketitle

\begin{abstract}
  We obtain polynomial-time approximation-preserving reductions (up to
  a factor of $1+\eps$) from the prize-collecting Steiner tree and
  prize-collecting Steiner forest problems in planar graphs to the
  corresponding problems in graphs of bounded treewidth. We also give
  an exact algorithm for the prize-collecting Steiner tree problem
  that runs in polynomial time for graphs of bounded treewidth. This,
  combined with our reductions, yields a PTAS for the prize-collecting
  Steiner tree problem in planar graphs and generalizes the PTAS of
  Borradaile, Klein and Mathieu \cite{BKM} for the Steiner tree
  problem in planar graphs.  Our results build upon the ideas in
  \cite{BKM} and the work of Bateni, Hajiaghayi and Marx
  \cite{BateniHM} on a PTAS for the Steiner forest problem in planar
  graphs. Our main technical result is on the properties of
  primal-dual algorithms for Steiner tree and forest problems in
  general graphs when they are run with scaled up penalties.
\end{abstract}

\section{Introduction}
\label{sec:intro}
The Steiner tree and Steiner forest problems are fundamental and
well-studied problems in network design. In the Steiner tree problem
we have an undirected graph $G=(V,E)$ with costs on the edges given by
$c \colon E \rightarrow \mathbb{R}^+$, and a set of terminals $S
\subseteq V$.  The goal is to find a minimum-cost tree in $G$ that
connects/contains the terminals. In the more general Steiner forest
problem we are given pairs of vertices $s_1t_1,\ldots, s_kt_k$ and the
goal is to find a minimum-cost forest which connects $s_i$ and $t_i$
for $1 \le i \le k$. Both problems are NP-Hard and APX-hard to
approximate.  These two problems have received considerable attention
in the approximation algorithms literature. For the Steiner tree
problem the very recent algorithm of Byrka \etal \cite{ByrkaGRS} gives a
$1.388$ approximation and this is the best known. For Steiner forest
problem the best known approximation is $2-1/k$ \cite{AKR}; this is
obtained via a natural cut-based LP relaxation. When $G$ is a planar
graph, Borradaile, Klein and Matheiu \cite{BKM} obtained a polynomial
time approximation scheme (PTAS)\footnote{A PTAS for a problem is an
  algorithm (technically speaking a family of algorithms), that for
  each fixed $\eps > 0$, gives a $(1+\eps)$-approximation in
  polynomial time.} for the Steiner tree problem. More recently,
Bateni, Hajiaghayi and Marx \cite{BateniHM} obtained a PTAS for the
Steiner forest problem.

In this paper we consider the prize-collecting versions of the above
problems, which have also received considerable attention.  In the
prize-collecting version of Steiner tree, we are given a root vertex
$r$ and non-negative penalties $\pi:V \rightarrow \mathbb{R}^+$ on the
vertices of $G$. The goal is to find a tree $T$ to minimize the sum of
the edge-cost of $T$ and the penalties of the vertices not included in
$T$; formally, the objective is to minimize $\sum_{e \in E(T)}c(e) +
\sum_{v \not \in V(T)} \pi(v)$. It is easy to see that the Steiner
tree problem is the special case with an infinite penalty on the
terminals and zero penalty on non-terminals. In the prize-collecting
Steiner forest problem we have a penalty $\pi(uv)$ for each pair of
vertices $u,v \in V$.  The goal is to find a forest $F\subseteq E$ to
minimize $\sum_{e \in F} c(e) + \sum_{\text{$uv$ not connected by
    $F$}} \pi(uv)$. For the prize-collecting Steiner tree problem, the
current best approximation ratio is $2-\delta$ for some small but fixed
$\delta > 0$, due to Archer \etal
\cite{ArcherBHK}, and for the prize-collecting Steiner forest problem,
the best known ratio is $2.54$ due to Hajiaghayi and Jain
\cite{HajiaghayiJ}.

In this paper we obtain a PTAS for the prize-collecting Steiner tree
problem in planar graphs; we obtain this result via an approximation
preserving reduction from the prize-collecting Steiner tree problem in
planar graphs to the problem in graphs of bounded treewidth.
Similarly for the prize-collecting Steiner forest problem, we reduce
the problem in planar graphs to the problem in graphs of bounded
treewidth. Below is our main result.

\begin{theorem}\label{thm:reduction}
  For each fixed $\eps > 0$, there is a $\rho(1+\eps)$-approximation
  for the prize-collecting Steiner forest problem in planar graphs if,
  for for each fixed integer $k$, there is a $\rho$-approximation for
  the prize-collecting Steiner forest problem in graphs of treewidth
  at most $k$.  In particular, there is a $\rho(1+\eps)$-approximation
  for the prize-collecting Steiner tree problem in planar graphs if
  there is a $\rho$-approximation for the prize-collecting Steiner
  tree problem in graphs of bounded treewidth.
\end{theorem}

We describe an exact algorithm for the prize-collecting Steiner tree
problem in graphs of bounded treewidth:

\begin{theorem}\label{thm:pcstBoundedTreewidth}
  For any fixed integer $k$, there is a polynomial-time algorithm that
  exactly solves the prize-collecting Steiner tree problem in graphs
  of treewidth at most $k$. 
\end{theorem}

Combining the reduction in Theorem~\ref{thm:reduction} 
with the above theorem, we obtain the following.

\begin{corollary}
  There is a PTAS for the prize-collecting Steiner tree problem in
  planar graphs.
\end{corollary}

A natural idea would be to show that there is a PTAS for the
prize-collecting Steiner forest problem in graphs of bounded
treewidth.  (We do not expect an exact algorithm since the problem is
NP-Hard even for this case \cite{BateniHM}.) This would yield a PTAS
for the problem in planar graphs. However, recent work in
\cite{BateniHM-pcst} shows that prize-collecting Steiner forest is
APX-Hard even in series-parallel graphs, which are planar and have
treewidth 2.

\begin{remark}
  Results analogous to Theorems~\ref{thm:reduction} and
  \ref{thm:pcstBoundedTreewidth} also apply to the prize-collecting
  Traveling Salesperson Problem; details are deferred to a later version
  of the paper.
\end{remark}

\begin{remark}
  The results for planar graphs extend to graphs of bounded
  genus. This follows previous ideas; details are deferred to a later
  version of the paper.
\end{remark}

We are motivated to consider the prize-collecting problems for several
reasons. First, they generalize the Steiner tree and Steiner forest
problems. Second, the prize-collecting Steiner tree problem has played
a crucial role in algorithms for $k$-MST \cite{Garg,AroraK},
$k$-Stroll \cite{ChaudhuriGRT} and the Orienteering
\cite{Orienteering,TimeWindows,ChekuriKP} problems. These problems are
more difficult than the Steiner tree problem because they also involve
choosing the subset of terminals to connect. In particular, obtaining
a PTAS for $k$-MST, $k$-Stroll or Orienteering in planar graphs
appears to be challenging with current techniques. The
prize-collecting version of Steiner tree can be viewed as an
intermediate problem that still captures some of these difficulties in
that the set of terminals to connect is not determined \textit{a
  priori}.  Techniques and ideas developed for the prize-collecting
problems are likely to play a role in making progress on problems such
as $k$-MST and Orienteering. We note that several of these problems
have PTASes in low-dimensional Euclidean spaces but those rely on
space partitioning schemes such as those of Arora \cite{Arora} and
Mitchell \cite{Mitchell}; planar graph PTASes for network design have
been more difficult to obtain since there is no comparable generic
scheme. Some of the technical challenges will be outlined below where
we describe our techniques at a high-level.

\subsection{Overview of Techniques}
Through this paper, we use $\opt$ to denote the cost of an optimal
solution to the given problem instance. Also, for any solution $F$ for
the prize-collecting Steiner forest (or prize-collecting Steiner
tree), we use $\len(F)$ to denote $\sum_{e \in F} c(e)$, $\pen(F)$ to
denote $\sum_{\text{$uv$ not connected by $F$}} \pi(uv)$, and
$\cost(F)$ to denote $\len(F) + \pen(F)$. For the prize-collecting
Steiner tree problem, we refer to all vertices with non-zero penalty
as \emph{terminals}; similarly, for prize-collecting Steiner forest,
we refer to all pairs of vertices with non-zero penalty as
\emph{terminal pairs}.
Finally, we assume that $\eps \le 1$ and (w.l.o.g.) that all terminals
in the input graph $G$ have degree 1; if this is not the case for some
terminal $v$, simply connect it to a new vertex $v'$ using an edge of
cost $0$, and use $v'$ as a terminal in the place of $v$.

Planar graph approximation schemes for Steiner tree \cite{BKM} and
forest \cite{BateniHM} build on many ideas starting with the
framework of Baker \cite{Baker} for PTASes for planar graphs and
utilizing several subsequent technical tools. Before describing the
new technical contributions needed to handle prize-collecting
problems, we first give a very high-level overview of the approach
common to our algorithms, and the previous approximation schemes of
\cite{BKM, BateniHM}.\footnote{All the algorithms modify this approach
  in distinct ways; see the subsequent discussion.}

\begin{enumerate}
\item Given a planar graph $G$, construct a \emph{spanner}, a subgraph
  $H \subseteq G$ such that there exists a solution in $H$ of cost
  $(1+\eps) \opt$ and the total length of edges in $H$ is $f(\eps)
  \opt$, for some function $f$ depending purely on $\eps$.

\item Partition the edges of $H$ into $f(\eps)/\eps$ sets $E_1, E_2,
  \ldots E_{f(\eps)/\eps)}$, such that contracting any set results in a
  graph of treewidth $O(f(\eps)/\eps)$ \cite{DemaineHM}. 

\item Pick the set $E_i$ of minimum total edge length; this cost must be
  no more than $\eps \opt$. Contract the edges of this set $E_i$,
  yielding a bounded treewidth graph that contains a solution of cost
  $(1+\eps)\opt$. Solve the problem on this graph of bounded
  treewidth. This may not correspond to a solution in the original
  graph $G$; add (uncontract) edges of $E_i$ as necessary to obtain a
  feasible solution in $G$. The total cost of $E_i$ is at most $\eps
  \opt$, and hence the cost of the solution is at most $(1+2 \eps)
  \opt$. 
\end{enumerate}

We note that the most difficult step is the first, finding a spanner
$H$. Once this is accomplished, the rest is somewhat standard;
applying a theorem of \cite{DemaineHM} allows one to decompose the
edge set of $H$ into pieces, and a simple averaging argument implies
that one of these pieces has low cost; contracting this piece yields a
graph of bounded treewidth. Depending on the problem being considered,
it may be possible to solve the bounded treewidth instance exactly; if
not, an approximation algorithm is used.  (In \cite{BKM}, steps 2 and
3 are modified slightly to obtain a more efficient algorithm; we omit
details here.) For prize-collecting Steiner tree, it is fairly
straightforward to solve the problem exactly in graphs of fixed
treewidth. As mentioned before, prize-collecting Steiner forest is
NP-Hard even in graphs of fixed treewidth; obtaining a small
constant-factor approximation would be of interest as, from
Theorem~\ref{thm:reduction}, this would immediately yield an algorithm
for planar graphs achieving a similar approximation ratio.
Thus, we focus below on Step 1, the construction of the spanner $H$;
we describe parts of the spanner constructions of \cite{BKM,BateniHM},
and the modifications necessary for the prize-collecting variants.

\smallskip
\noindent
{\bf Prize-collecting Steiner Tree:}
We first focus on the easier case of Steiner trees and the scheme in
\cite{BKM}. The basic idea to find a spanner for the Steiner tree
instance is as follows. (We assume we are given an embedding of the
input planar graph $G$.) The algorithm starts by computing a
$2$-approximate (any constant factor would do) Steiner tree $T$ in
$G$. One can use an Euler tour of $T$ and splice open along the tour
to obtain another plane graph $G'$ in which the Euler tour of $T$ is
its outer face; note that the total cost of edges on this face is at
most $4\opt$. Now, all the terminals are on the \emph{outer face} of
$G'$; this fact is crucial. The algorithm in \cite{BKM} then builds on
the work of Klein \cite{Klein} to obtain the desired spanner $H$: It
begins with the outer face of $G'$ (containing all the terminals), and
adds edges of $G$ of total cost proportional to the length of this
outer face, while guaranteeing the existence of a near-optimal
solution only using these added edges. We provide details of the
construction in Appendix~\ref{app:spanner}.

The main difficulty in extending the above PTAS to the
prize-collecting Steiner tree problem is the following. We again wish
to find a spanner subgraph $H$ of $G$ by starting with an
$O(1)$-approximate tree $T$ and making it the outer face. Unlike the
Steiner tree case we run into a difficulty. The approximate tree $T$
is not guaranteed to contain all the vertices that are connected to
the root in an optimal solution! In fact, if we knew the vertices that
need to be connected to the root in a near-optimal solution then we
can simply reduce the problem to the Steiner tree problem. We overcome
this difficulty by proving the following.

\begin{theorem}
  \label{thm:intro-tree}
  There is a polynomial time algorithm that, given a prize-collecting
  Steiner tree instance in a graph $G$, outputs a tree $T$ of cost
  $O(1/\eps) \opt$ such that there is a $(1+\eps)$-approximate
  solution $T'$ that connects to the root a subset of the vertices in
  $T$.
\end{theorem}

The algorithm to achieve the above is in fact the Goemans-Williamson
primal-dual algorithm for the prize-collecting Steiner tree problem
but with modified potentials $\pi'(v) = \frac{2}{\eps}\pi(v)$ for each
$v \in V$. By exploiting properties of the primal-dual algorithm we
can prove the above theorem. We then proceed as in the Steiner tree
case, making this tree $T$ the outer face, and adding edges to form a
spanner. We note that we cannot use an (approximation) algorithm for
the prize-collecting Steiner tree problem as a black box in proving
the above --- it is important to rely on the properties of the
primal-dual algorithm as we change the potentials.

\smallskip
\noindent {\bf Prize-collecting Steiner Forest:} The PTAS for Steiner
forest, which is quite recent \cite{BateniHM}, requires two new
ideas. First, as with the Steiner tree PTAS one starts with an
$O(1)$-approximate Steiner forest. However, this forest has
potentially many components and one cannot apply the spanner
construction idea of \cite{BKM} in making a single tree the outer
face.  At the same time one cannot directly argue that one can treat
each tree in the forest separately; optimal (or near-optimal)
solutions may connect vertices in different components of the forest.
In \cite{BateniHM} there is an additional key step in which the trees
are grown via a primal-dual type argument and some of them are merged
--- after this step, the remaining trees are ``far apart'' and hence
can be treated independently via the spanner approach of
\cite{BKM}. (More precisely, one obtains a collection of subgraphs
$H_1, H_2, \ldots$ such that $\sum_i \len(H_i) \le f(\eps)
\opt$. Further, if $\opt_i$ denotes the cost of an optimal solution
for the terminal pairs in $H_i$, then $\sum_i \opt_i \le (1+\eps)
\opt$.)

A second difficulty in obtaining a PTAS is in step 3, solving the
problem on bounded treewidth graphs. As mentioned previously, a PTAS
is developed in \cite{BateniHM} for Steiner forest in bounded
treewidth graphs.

As for the prize-collecting Steiner tree problem, the difficulty in
the prize-collecting Steiner forest problem is that we do not know
which terminal pairs to connect. We deal with this using an approach
similar to that for the prize-collecting Steiner tree problem.
However, the $3$-approximate primal-dual algorithm for the Steiner
forest problem, due to Hajiaghayi and Jain \cite{HajiaghayiJ}, is
quite complex; each step requires $O(n)$ max-flow computations.  It is
unclear whether one can prove a theorem similar to
Theorem~\ref{thm:intro-tree} via the algorithm in \cite{HajiaghayiJ}.
We develop a simpler primal-dual algorithm and analysis for the
Steiner forest problem that gives a $4$-approximation. We use the
structure and analysis of our algorithm to prove the following
theorem.

\begin{theorem}
  \label{thm:intro-forest}
  There is a polynomial time algorithm that, given a prize-collecting
  Steiner forest instance in a graph $G$, outputs a forest $F$ of cost
  $O(1/\eps) \opt$ such that there is a $(1+\eps)$-approximate
  solution $F'$ that connects a subset of the pairs connected by $F$.
\end{theorem}

\medskip
\noindent {\bf Other related work:} There is a substantial amount of
literature on various aspects related to the work discussed in this
paper.  We refer the reader to \cite{DemaineH08} for an overview of
the progress on obtaining approximation schemes for optimization
problems on planar graphs; recent papers \cite{Klein,BKM,BateniHM}
have pushed the techniques further for network design problems.
Prize-collecting versions of network design have received considerable
attention following the work of Goemans and Williamson \cite{GW} on a
primal-dual algorithm which had several applications to other problems
such as $k$-MST. Following Hajiaghayi and Jain's work on
prize-collecting Steiner forest problem \cite{HajiaghayiJ}, there have
been several other papers on prize-collecting network design
problems. Sharma, Swamy and Williamson \cite{SharmaSW07} generalized
the approach in \cite{HajiaghayiJ} to obtain primal-dual constant
factor approximation algorithms for prize-collecting constrained
forest problems (in the framework of Goemans and Williamson \cite{GW})
with submodular penalty functions. Gutner \cite{Gutner08} gave a very
simple and efficient local-ratio based $3$-approximation for prize
collecting Steiner forest problem which also applies to the
generalized version with more than two terminals in a group; although
Gutner's algorithm is quite simple, it does not seem possible to prove
Theorems~\ref{thm:intro-tree} and \ref{thm:intro-forest} directly via
his algorithm. Constant factor approximation ratios have also been
obtained for more general problems with higher connectivity
requirements \cite{NagarajanSW08,HajiaghayiKKN10}.

\medskip
\noindent {\bf Relation to independent work in \cite{BateniHM-pcst}:}
Bateni, Hajiaghayi and Marx \cite{BateniHM-pcst} have obtained several
results on prize-collecting network design problems in planar
graphs. Our work was done independently before we were aware of their
paper; we have since learnt that their results were obtained slightly
earlier than ours. We briefly compare the results and techniques
between their work and ours.  They have an analogous theorem to
Theorem~\ref{thm:reduction} but they also consider a more general
problem, namely prize-collecting Steiner forest with submodular
penalty functions. They obtain PTASes for prize-collecting Steiner
tree and related problems such as prize-collecting TSP and stroll
problems, much as we can obtain via the reduction to bounded treewidth
instances. They complement their algorithms with an interesting
hardness of approximation result which shows that prize-collecting
Steiner forest is APX-hard in series-parallel graphs which are planar
graphs and have treewidth $2$. Their APX hardness also extends to
instances in the Euclidean plane. The main part of proof of the
reduction to bounded treewidth instances in \cite{BateniHM-pcst}
differs from ours in some ways. While we rely on properties of a
primal-dual algorithm for the underlying problem with scaled up
penalties, they use a separate primal-dual clustering step on top of
the trees returned by an approximation algorithm for the underlying
problem; this is inspired by earlier work of Archer \etal
\cite{ArcherBHK} and further extended in \cite{BateniHM}.

\medskip
\noindent
{\bf Outline:}
In Section~\ref{sec:primalDual}, we give a $4$-approximate primal-dual
algorithm for the prize-collecting Steiner forest problem. In
Section~\ref{sec:reduction}, we prove Theorem~\ref{thm:reduction},
giving the complete reduction from the prize-collecting problems in
planar graphs to their bounded-treewidth versions. This reduction has
three parts: First, in Section~\ref{subsec:scaled} we prove
Theorems~\ref{thm:intro-tree} and \ref{thm:intro-forest}, showing that
we can get $O(1/\eps)$-approximate solutions connecting almost all the
terminals connected by an optimal solution. Second, we use these
theorems to complete the spanner constructions; details are provided
in Appendix~\ref{app:spanner}. Third, we give the remaining details of
the reduction in Section~\ref{subsec:contraction}.  Finally, in
Section~\ref{sec:treewidth}, we give an exact algorithm for the
prize-collecting Steiner tree problem in graphs of fixed treewidth.

\section{A Primal-Dual Algorithm for Prize-Collecting Steiner Forest}
\label{sec:primalDual}

Our algorithm is similar to the Goemans-Williamson \cite{GW}
primal-dual algorithm for prize-collecting Steiner tree. For the
prize-collecting Steiner tree, this gives a $2$-approximation (with some
additional guarantees), but this does not appear possible in the
Steiner forest variant. In this section, we give a simple
$4$-approximation for the prize-collecting Steiner forest problem, and
later exploit properties of this algorithm to prove
Theorem~\ref{thm:intro-forest}.

We give primal and dual linear programming formulations for the
prize-collecting Steiner forest problem below. For each pair $(s_i,
t_i)$ the variable $z_i$ is $1$ if we pay the penalty for not
connecting the pair, and $0$ otherwise; the variable $x_e$ denotes
whether the edge $e$ is selected for the forest. We abuse notation and
say that a set $S$ separates (the terminal pair) $i$ if it separates
$s_i$ from $t_i$. Let $\S_i$ denote the collection of sets $S$ that
separate $i$.

\setlength{\columnsep}{0.1in} \setlength{\columnseprule}{0.3pt}
\begin{multicols}{2}
  \begin{center} \textbf{Primal-PCSF}
  \end{center}
  \begin{align*} 
      \min \sum_{e} c_e x_e & + \sum_i \pi_i z_i & \\
     \sum_{e \in \delta(S)} x_e \ \ \ge & \quad (1 - z_i) & (\forall
     i, S \in \S_i) \\
     x_e, z_i \ \ \ge & \quad 0 & (\forall e, i)\\
  \end{align*}

  \begin{center} \textbf{Dual-PCSF}
  \end{center}
  \vspace{-0.15in}
  \begin{align*} 
    \max \sum_S \sum_{i \in \S_i} & y_{i,S} &  \\
    \sum_{S \colon e \in \delta(S)} \sum_{i \colon S \in \S_i}
    y_{i,S} \ \le & \quad c_e & (\forall e) \\
    \sum_{S \in \S_i} y_{i,S} \ \le & \quad \pi_i &
    (\forall i)\\
    y_{i,S} \ \ge & \quad 0 & (\forall i, S \in \S_i)
  \end{align*}
\end{multicols}

A reader unfamiliar with the primal-dual algorithms of \cite{AKR,GW}
(for the Steiner forest and prize-collecting Steiner tree problems
respectively) may wish to skip this paragraph and the next, proceeding
directly to the description of our algorithm. Here, we briefly
describe one way in which our algorithm differs from those of
\cite{AKR,GW}.
A natural approach is to initially assign each terminal $s_i$ or $t_i$
a potential of approximately $\pi_i$, and make each one an active
component. In both the Steiner forest and prize-collecting Steiner
tree problems, a natural LP formulation has a single dual variable
$y_S$ for each active component $S$. One increases the $y_S$ variable
for each active component $S$ uniformly, until either a component
``runs out'' of potential, or an edge becomes ``tight''. In the former
case, the component is deactivated; in the latter, one merges the two
components adjacent to the edge, and combines their potentials.

In the prize-collecting Steiner forest problem, however, we have a
\emph{collection} of dual variables $y_{i,S}$ for a single active
component $S$. We still wish each active component $S$ to grow at a
uniform rate, but it is now not clear which pair $(s_i, t_i)$
separated by $S$ should pay for the growth of $S$. Below, we show that
there is a natural way to share the cost of this growth among the
pairs separated by $S$; once this is done, the analysis proceeds as in
\cite{AKR,GW}. 

Before describing the algorithm, we present some useful notation.  The
primal-dual algorithm begins with a growth phase, followed by a
deletion phase. At all times, our algorithm maintains a forest $F$ of
edges. The connected components of $F$ are labeled either
\emph{active} or inactive; we use $\C$ to denote the set of active
components at any given time. Every component is active at the time it
is formed; we use $\hat{\C}$ to denote the set of components that were
ever formed during the algorithm's execution. (Note that as we only
add edges to $F$ during the growth phase, $\hat{\C}$ is a laminar
family of components.) When a component $S \in \hat{\C}$ is formed, we
assign it a potential $\pot(S)$ that corresponds (roughly) to the
penalty of terminal pairs separated by $S$; thus, this measures how
much we are willing to pay in order to connect terminals in $S$ to
their partners outside. If we grow $S$ by more than its potential
without meeting the desired partners, then it is more effective to pay
the penalty than connect the pairs separated by $S$; at this point, we
will mark $S$ inactive. (Of course, some terminals in $S$ may meet
their partners while others do not; this is the key difference between
Steiner tree and Steiner forest, and we describe how to make this
intuition more precise below.)

We form new components by adding edges to $F$, merging existing
components; when two components merge, they combine their potentials.
We say that $S \in \hat{\C}$ \emph{unites} $i$ if $S$ is the smallest
active component in $\hat{\C}$ containing both $s_i$ and
$t_i$. Terminals are labeled \emph{satisfied}, \emph{alive}, or
\emph{dead}; initially, all terminals are alive. If a component $S$
uniting $i$ is formed, the terminals $s_i$ and $t_i$ are marked
satisfied if they were both alive immediately before the formation of
$S$. (As satisfied terminals are no longer willing to pay for the
growth of a component $S$ containing them, we adjust $\pot(S)$ if
necessary; the procedure {\sc ProcessHistory} handles the necessary
bookkeeping.)  Once we form a component $S \in \hat{\C}$, we ``grow''
it by increasing an auxiliary dual variable $y(S)$; simultaneously, we
decrease $\pot(S)$ to pay for this growth. Once $\pot(S)$ becomes $0$,
a component is labeled inactive. (Recall that all components are
active when they are formed.) When a component becomes inactive, all
of its unsatisfied terminals are marked dead. (Intuitively, it is now
more effective to pay the associated penalties than to try to connect
them to their partners.)  For any terminal $s_i$ (respectively $t_i$),
we use $\hist(s_i)$ (respectively $\hist(t_i)$) to denote the set of
components $S \in \hat{\C}$ such that $S$ contains $s_i$ ($t_i$) and
$s_i$ ($t_i$) was alive after $S$ was formed.

Finally, we note that the auxiliary variables $y(S)$ do not exist in
our dual LP formulation; instead, we have variables $y_{i,S}$. We
ensure that at the end of the algorithm, $y(S) = \sum_{i \colon S \in
  \S_i} y_{i,S}$. In order to determine how $y(S)$ is split among the
variables $y_{i,S}$ we maintain an associated variable $\unc(S)$, the
uncharged growth of $S$. The following proposition is entirely
straightforward, as we increase $y(S)$ and $\unc(S)$ together:

\begin{proposition}
  For each $S \in \hat{\C}$, as soon as $S$ is marked inactive or $S$
  becomes part of a larger component, we have $\unc(S) = y(S)$.
\end{proposition}

The procedure {\sc ProcessHistory} ensures that $\unc(S)$ (and hence
$y(S)$) is split appropriately among the dual variables $y_{i,S}$. 
We omit a proof of the proposition below:

\begin{proposition}
  When the main While loop of {\sc Primal-Dual Forest} terminates, for
  each pair of terminals $(s_i, t_i)$, we have called {\sc
    ProcessHistory}($s_i$) and {\sc ProcessHistory}($t_i$).
\end{proposition}

We initialize the set of active components to be the set of terminals,
and set $\pot(s_i) = \pot(t_i) = \pi_i/2$ for each $1 \le i \le h$.

\begin{algo}
\underline{\sc Primal-Dual Forest} \\
$F \leftarrow \emptyset$ \\
$\C \leftarrow \{ \{s_i\} \colon 1 \le i \le h \} \cup 
  \{ \{t_i\} \colon 1 \le i \le h \} $\\
For all $1 \le i \le h$: \+ \\
  $\pot(\{s_i\}) = \pot(\{t_i\}) = \pi_i / 2$ \- \\
While there is an active component: 
           \hspace{2.5in}\Comment{Begin Main Loop} \+ \\
  Increase $y_S$ by $\varDelta$ for each $S \in \C$ until
  an edge goes tight, or a component uses all its potential. \\
  For each $S \in \C$: \+ \\
    $\pot(S) \leftarrow \pot(S) - \varDelta$; $\unc(S) \leftarrow \unc(S)
    + \varDelta$ \- \\
  If (edge $e$ connecting $S_1,S_2$ goes tight) \+ \\
    $F \leftarrow F \cup \{e\}$ \\
    $\C \leftarrow \left(\C \setminus \{S_1, S_2\} \right) \cup \{S_1
      \cup S_2\}$ \\
    $y(S_1 \cup S_2) \leftarrow 0$; $\unc(S_1 \cup S_2) \leftarrow 0$ \\
    $\pot(S_1 \cup S_2) \leftarrow \pot(S_1) + \pot(S_2)$ \\
    For each $i$ such that $S_1 \cup S_2$ unites $i$: \+ \\
      If ($s_i$ is alive), $\pot(S_1 \cup S_2) \leftarrow \pot(S_1
      \cup S_2) - $ {\sc ProcessHistory}($s_i$)\\ 
      If ($t_i$ is alive), $\pot(S_1 \cup S_2) \leftarrow \pot(S_1 \cup S_2) - $ {\sc
        ProcessHistory}($t_i$) \\
      If ($s_i$ AND $t_i$ are alive), Mark $s_i, t_i$ as satisfied \\
      Else, Mark $s_i, t_i$ as dead \- \- \\
  If (component $S$ uses all its potential) \+ \\
    $\C \leftarrow \C \setminus \{S\}$ \hspace{3.33in} \Comment{Mark S as inactive} \\ 
    For each alive $s_i \in S$ : \hspace{2.02in} \Comment{Similarly for
      each alive $t_j \in S$}\+ \\
      Mark $s_i$ as dead \\
      {\sc ProcessHistory}($s_i$) \- \- \- \\
End While \hspace{4.15in}\Comment{End Main Loop}\\
For each edge $e \in F$: \hspace{3.53in}\Comment{Deletion Phase}\+ \\
  Delete $e$ if $F - e$ does not separate any pair of satisfied terminals \-
\end{algo}

\begin{algo}
\underline{{\sc ProcessHistory($s_i$)}:}\\
$\rem(s_i) \leftarrow \pi_i / 2$ \\
For each $S \in \hist(s_i)$, in increasing order of size: \+ \\
  If $\unc(S) \le \rem(s_i)$: \hspace{0.1in} \Comment{All uncharged growth of S can
    be charged to $s_i$}\+ \\
    $\rem(s_i) \leftarrow \rem(s_i) - \unc(S)$ \\
    $y_{i,S} \leftarrow \unc(S)$ \\
    $\unc(S) \leftarrow 0$ \- \\
  Else: \hspace{3.25in} \Comment{Charge as much as possible to $s_i$} \+ \\
    $\unc(S) \leftarrow \unc(S) - \rem(s_i)$ \\
    $y_{i,S} \leftarrow \rem(s_i)$ \\
    $\rem(s_i) \leftarrow 0$ \\
    Return 0 \- \- \\
Return $\rem(s_i)$
\end{algo}

\begin{lemma}\label{lem:invariant}
  For each component $S \in \hat{\C}$, $y(S) = \sum_{i \colon S \in
    \S_i} y_{i,S}$. For each dead terminal $s_i$, we have $\sum_{S \in
    \hist(s_i)} y_{i,S} = \pi_i / 2$; similarly, for each dead $t_i$,
  $\sum_{S \in \hist(t_i)} y_{i,S} = \pi_i / 2$.
\end{lemma}
\begin{proofsketch}
  It is easy to verify these statements by checking that the algorithm
  maintains the invariant that for each component $S \in \hat{\C}$,
  $\pot(S) + \sum_{S' \in \hat{\C} : S' \subseteq S}\unc(S') =
  \sum_{\textrm{Alive } i \in S} \pi_i / 2$.
\end{proofsketch}

\begin{lemma}\label{lem:feasibleDual}
  The variables $y_{i,S}$ from the algorithm {\sc Primal-Dual Forest}
  correspond to a feasible dual solution for the LP \textbf{Dual-PCSF}.
\end{lemma}
\begin{proof}
  It is easy to verify that all constraints are satisfied: For any
  edge $e$, it is added to $F$ once $\sum_{S: e \in \delta(S)} y(S) =
  c_e$, and subsequently there is no active component $S$ such that $e
  \in \delta(S)$. As $y(S) = \sum_{i \colon S \in \S_i}
  y_{i,S}$, we satisfy the associated constraint. Further, for any
  pair $s_i, t_i$, the procedure {\sc ProcessHistory} guarantees that
  $\sum_{S: s_i \in S, t_i \not \in S} y_{i,S} \le \pi_i/2$, and similarly
  $\sum_{S: s_i \not \in S, t_i  \in S} y_{i,S} \le \pi_i/2$.
\end{proof}

\begin{theorem}\label{thm:pcsf-4}
  If $F$ denotes the forest returned by the algorithm {\sc Primal-Dual
    Forest}, then $\sum_{e \in F} c_e + \sum_{i \textrm{ separated by
    } F} \pi_i \le 4 \opt$, where $\opt$ denotes the cost of an
  optimal prize-collecting Steiner forest.
\end{theorem}
\begin{proof}
  As $\opt$ is upper bounded by the value of any feasible dual
  solution, it suffices to show 

  \[\sum_{e \in F} c_e + \sum_{i \textrm{ separated by } F} \pi_i \le
  4 \sum_{S} \sum_{i \colon S \in \S_i} y_{i,S}.\]

  For any $i$ to be separated by the forest $F$, either $s_i$ or $t_i$
  (or both) must have been marked dead, and thus from
  Lemma~\ref{lem:invariant}, we have $\pi_i \le 2 \sum_{S \in \S_i}
  y_{i,S}$. Hence it suffices to prove
  \[\sum_{e \in F} c_e + \sum_{i \textrm{ separated by } F} 2 \sum_{S
    \in \S_i} y_{i,S}\le \left(2 \sum_{S} \sum_{i \colon S \in \S_i} y_{i,S}
  \right) + \left(2 \sum_i \sum_{S \in \S_i} y_{i,S}\right).\]

  We prove $\sum_{e \in F} c_e \le 2 \sum_{S} y(S)$, which, using the
  fact that $y(S) = \sum_{i \colon S \in \S_i} y_{i,S}$, implies the
  inequality above.
  One can now use the standard primal-dual proof technique of
  \cite{AKR, GW}. Since an edge $e$ is added to $F$ only when it
  becomes tight, we have $c_e = \sum_{S: e \in \delta(S)} y(S)$;
  hence, the desired inequality is equivalent to:
  \[\sum_{S} |\delta(S) \cap F| y(S) \le 2 \sum_S y(S).\]

  To verify this inequality, we check that in every iteration of the
  while loop, the \emph{increase} in both the left- and right-hand
  sides satisfies the inequality. Since $y(S)$ increases only for
  components that are active in a given iteration, and we raise $y(S)$
  uniformly by $\varDelta$ for each $S \in \C$, this is equivalent to
  checking that in any iteration, $\sum_{\textrm{Active }S} |\delta(S)
  \cap F| \varDelta \le 2 \varDelta \cdot n_{a}$, where $n_a$ denotes
  the number of components active in this iteration. 

  Construct an auxiliary graph $H$ by beginning with $G(V,F)$, and
  shrinking each currently active and inactive component to a single
  vertex. Discard any isolated vertex of $H$ corresponding to an
  inactive component. We now argue that the average degree in $H$ of
  the vertices corresponding to active components is at most $2$; this
  completes the proof. To bound the average degree of the vertices
  corresponding to active components, we note that the average degree
  of all vertices is less than $2$ (as $H$ is a forest), and that each
  vertex corresponding to an inactive component has degree at least
  $2$. If the latter were not true, there is an inactive component
  which is a leaf of $H$, incident to a single edge $e$. But no
  inactive component separates a satisfied terminal from its partner,
  and so the edge $e$ would have been deleted from $F$ in the Deletion
  Phase of {\sc Primal-Dual Forest}. 
\end{proof}

\paragraph{Running time:} {\sc Prima-Dual Forest} can be implemented
in $O(n h + n^2 \log n)$ time where $h$ is the number of pairs with
strictly positive penalties and $n$ is the number of nodes in the
graph. The only additional difficulty in the algorithm as compared to
the primal-dual algorithmic framework of Goemans and Williamson
\cite{GW} for prize-collecting Steiner tree and related problems is
the {\sc ProcessHistory} step. When two active components merge, for
each pair $s_it_i$ united by this merge, we have to go over the sets
in $\hist(s_i)$ and $\hist(t_i)$. The sets in $\hat{\C}$ form a
laminar family on $V$ and hence $|\hat{\C}| = O(n)$ and thus {\sc
  ProcessHistory}($s_i$) can be implemented in $O(n)$ time by
considering each set in $|\hat{\C}|$. Similarly, {\sc
  ProcessHistory}($t_i$). Since a pair is united at most once, and
each terminal is marked dead at most once, the total work involved in
processing the histories is $O(n h)$.  Priority queues can be used to
keep track of the events corresponding to edges becoming tight and
components becoming inactive; this is very similar to the
implementation in \cite{GW} and the total work involved is $O(n^2 \log
n)$ time.

\section{The Reduction to Bounded Treewidth: Building Spanners}
\label{sec:reduction}

Recall that the first step in building a spanner for the Steiner tree
and forest problems was to construct a new plane graph $G'$ in which
(i) all terminals are on the outer face and (ii) the length of the
outer face is $O(1) \cdot \opt$. This was done by splicing open an
Euler tour of an $O(1)$ approximate solution in the original graph
$G$, and converting the tour into the outer face.\footnote{In the case
  of Steiner forest, this is done separately for distinct trees in the
  $O(1)$-approximate forest.} In the prize-collecting versions,
however, we do not know which terminals to connect, and it is not
possible to find an $O(1)$-approximate solution in which \emph{all}
the terminals are connected to the root (for Steiner tree) or to their
partners (for Steiner forest). 

In Section~\ref{subsec:scaled}, we prove Theorems~\ref{thm:intro-tree}
and \ref{thm:intro-forest}, showing that we can find a solution of
cost $O(1/\eps) \opt$ which connects ``almost'' all the terminals
connected by any optimal solution. More precisely, the total penalty
of terminals connected by an optimal solution but not by our
$O(1/\eps)$-approximate solution is at most $\eps \opt$. In
Appendix~\ref{app:spanner}, we complete the construction of the
spanner, using ideas from \cite{BKM, BateniHM}.

\subsection{Scaling Penalties to Capture Important Terminals }
\label{subsec:scaled}

We prove Theorem~\ref{thm:intro-forest}, which implies
Theorem~\ref{thm:intro-tree}, as prize-collecting Steiner tree is a
special case of prize-collecting Steiner forest. 
Given an instance $I$ of the prize-collecting Steiner forest on a
graph $G(V,E)$, with $c_e$ denoting the cost of edge $e \in E$ and
$\pi_i$ the penalty for not connecting $s_i$ to $t_i$, we define a new
instance $I'$ as follows: The graph and edge cost functions are
unchanged, but we scale the penalties so that the penalty for not
connecting $s_i$ to $t_i$ is $\pi'_i = 2 \pi_i / \eps$.

\begin{theorem}\label{thm:scaledPenalties}
  Let $F^*$ be \emph{any} optimal solution to an instance $I$ of
  prize-collecting Steiner forest, and let $\opt = \sum_{e \in F^*}
  c_e + \sum_{i \textrm{ separated by } F^*} \pi_i$. Let $F'$ be the
  forest output by algorithm {\sc Primal-Dual Forest} on the instance
  $I'$ with penalties scaled as above. Let $X$ denote the index set of
  the terminal pairs separated by $F'$ but not by $F^*$.  Then,
  $\sum_{e \in F'} c_e \le 8 \opt / \eps$, and $\sum_{i \in X} \pi_i
  \le \eps \opt$.
\end{theorem}
\begin{proof}
  We first note that the cost of an optimal solution to $I'$ is at
  most $2 \opt/\eps$; simply use the forest $F^*$, which pays $2/\eps$
  times as much penalty for every separated pair as it did in
  $I$. Thus, as {\sc Primal-Dual Forest} is a $4$-approximation, we
  have $\sum_{e \in F'} c_e \le 8 \opt/\eps$.

  To prove that the total penalty of pairs in $X$ is small, consider a
  Steiner forest instance defined on these pairs: As $F^*$ connects
  all the terminals in $X$ to their partners, the cost of an optimal
  Steiner forest for $X$ is at most $\opt$. Suppose, by way of
  contradiction, that $\sum_{i \in X} \pi_i > \eps \opt$, and hence
  that $\sum_{i \in X} \pi'_i > 2 \opt$.  Now consider the following
  dual of a natural LP for the Steiner Forest instance induced by $X$:

  \begin{center} \textbf{Dual-Steiner Forest($X$)}\\
    \vspace{-0.2in}
    \[\hspace{-1in} \max  \sum_{S \textrm{ separating some }i \in X}  z_S \]
    \vspace{-0.2in}
    \begin{align*} 
      \sum_{S \colon e \in \delta(S)} z_S \ \le & \quad c_e & (\forall e) \\
      z_S \ \ge & \quad 0 & (\forall S)
    \end{align*}
  \end{center}

  Let $y_{i,S}$ be the feasible solution to \textbf{Dual-PCSF}
  returned by {\sc Primal-Dual Forest} on instance $I'$.  Now,
  construct a dual solution to the LP \textbf{Dual-Steiner
    Forest($X$)} as follows: For each set $S$ separating some pair
  $s_it_i$ with $i \in X$, set $z_S = \sum_{i \in X} y_{i,S}$ .  As
  $\sum_{S : e \in \delta(S)} \sum_{i \colon S \in \S_i} y_{i,S} \le
  c_e$ from the feasibility of the solution to \textbf{Dual-PCSF}, we
  conclude that the dual variables $z_S$ correspond to a feasible
  solution of \textbf{Dual-Steiner Forest($X$)}.
  
  Thus, we have a feasible solution to \textbf{Dual-Steiner
    Forest($X$)} of total value $\sum_{S} \sum_{i \in X \colon S \in
    \S_i} y_{i,S}$. But each $i \in X$ was not connected by $F'$, and
  so we must have marked either $s_i$ or $t_i$ as dead. Hence, from
  Lemma~\ref{lem:invariant} $\sum_{S \in \S_i} y_{i,S} \ge
  \pi'_i/2$. That is, the value of our feasible dual solution is at
  least $\sum_{i \in X} \pi'_i / 2 > \opt$. By weak duality, the
  length of any Steiner forest for $X$ must be greater than
  $\opt$. But $F^*$ is a Steiner forest for $X$ of total length at
  most $\opt$, which is a contradiction.
\end{proof}

\bigskip
\noindent We can now prove Theorem~\ref{thm:intro-forest}:

\begin{proofof}{Theorem~\ref{thm:intro-forest}}
  Let $F^*$ be an optimal solution to a given instance $I$ of Steiner
  forest, and let $\opt = \sum_{e \in F^*} c_e + \sum_{i \textrm{
      separated by } F^*} \pi_i$. Construct a forest $F$ by running
  algorithm {\sc Primal-Dual Forest} on the scaled instance $I'$; from
  Theorem~\ref{thm:scaledPenalties} above, the total length of edges
  in $F$ is at most $8\opt/\eps$. If $X$ denotes the terminal pairs
  separated by $F$ but not by $F^*$, the penalty paid by $F$ is at
  most $\sum_{i \text{ separated by } F^*} \pi_i + \sum_{i \in X} \pi_i \le
  \opt + \eps \opt$. Thus, $F$ is a forest of total cost $O(8/\eps + (1
  + \eps))  \opt$.

  It remains only to argue that there is a $(1+\eps)$-approximate
  solution $F'$ that connects a subset of the pairs connected by
  $F$. Let $F^*_{-X}$ denote a solution formed from $F^*$ by paying
  the penalty for any terminal pair in $X$; clearly, the cost of
  $F^*_{-X}$ is the cost of $X$ added to $\sum_{i \in X} \pi_i$, which
  is at most $\opt + \eps \opt$.
\end{proofof}

\bigskip
The first part of the spanner construction for the prize-collecting
Steiner tree problem is now complete: As guaranteed by
Theorem~\ref{thm:intro-tree}, find a tree $T$ of cost $O(1/\eps) \opt$
such that there exists a $(1+\eps)$-approximate solution only
connecting terminals in $T$. Now form an Euler tour of $T$ by
duplicating edges, splice along this tour, and make the tour the outer
face of a new graph $G'$. Now, we have a graph in which all (relevant)
terminals are on the outer face, and the total length of the outer
face is $O(1/\eps) \opt$. The rest of the spanner construction
proceeds along the lines of \cite{BKM}; see
Appendix~\ref{app:spanner}.

For the prize-collecting Steiner forest problem, however, more work is
required. The forest $F$ guaranteed by Theorem~\ref{thm:intro-forest}
may have many components, which cannot be treated in isolation. As in
\cite{BateniHM}, we use a {\em prize-collecting} clustering scheme to
merge some components of the forest. The intuition is that after this
clustering, the remaining components are ``far apart'', and hence can
be treated separately. The prize-collecting clustering algorithm is as
follows: Contract each tree $T_i$ of $F$ to a single vertex $v_i$, to
obtain a new graph $\hat{G}$. We let $v_1, v_2, \ldots$ denote the
vertices of $\hat{G}$ corresponding to the contracted trees $T_1, T_2,
\ldots$ of $F$; note that $\hat{G}$ additionally has the vertices of
$G$ not contained in any tree $T_i$.  With each $v_i$, we associate a
potential $\phi_{v_i} = \frac{1}{\eps} \len(T_i)$.  Now run the
standard prize-collecting primal-dual algorithm as in \cite{GW}, using
potentials $\phi_{v_i}$. Initially, each vertex is an active component,
with potential $\phi_{v_i}$; in each step, the algorithm decreases
potentials of all active components uniformly until either a component
runs out of potential, or an edge becomes ``tight''. In the former
case, the component is marked as inactive. In the latter case, the two
components adjacent to the edge are merged, and their potentials
combined.  This is similar to the algorithm {\sc Primal-Dual Forest}
in Section~\ref{sec:primalDual}, but simpler, as we do not have the
additional accounting necessary to handle terminal pairs that only
wish to connect to each other. As in {\sc Primal-Dual Forest}, the
algorithm maintains a collection of dual variables $y_{v_i,S}$ for each
$v_i$ and set $S \subseteq V(\hat{G})$. A complete
description of this algorithm is given in \cite{BateniHM}; we omit
details from this paper.

The first stage of the clustering algorithm terminates when all
components are marked inactive. Let $F_1$ denote the forest of tight
edges selected by the algorithm after the first stage.  In the second
stage, we delete any edge $e$ from $F_1$ if it is the unique edge
incident to an inactive component. Let $F_2$ denote the set of edges
remaining. 

\begin{lemma}[\cite{BateniHM}]\label{lem:pcCost}
  The total length of all edges in $F_2$ is at most $2 \sum_{i} \phi_{v_i} = \frac{2}{\eps} \cdot \len(F)$, which is
  $O(\frac{1}{\eps^2}) \opt$.
\end{lemma}

We also use the following two technical lemmas. A graph $\hat{H}
\subseteq \hat{G}$ is said to \emph{exhaust} a vertex $v_i \in
V(\hat{G})$ if, for all $S$ such that $y_{v_i,S} > 0$, $\hat{H}$
contains at least one edge of $\delta(S)$. (Recall that $v_i$
corresponds to the contracted tree $T_i$ of $F$.)

\begin{lemma}[Lemma 10 of \cite{BateniHM}]\label{lem:mergingComponents}
  Let $V'$ be the set of vertices of $\hat{G}$ exhausted by a graph
  $H$. $\len(H) \ge \sum_{v_i \in V'
} \phi_{v_i}$.
\end{lemma}

\begin{lemma}[\cite{BateniHM}]\label{lem:exhaust}
  Let $\hat{H} \subseteq \hat{G}$ connect two vertices $v_1, v_2$ in
  distinct components of $F_2$. Then, $\hat{H}$ exhausts at least one
  of $v_1, v_2$.
\end{lemma}

Let $\T^1, \T^2, \ldots $ be the trees comprising the forest $F_2$.
We will now argue that these trees are sufficiently ``far apart'', and
so we can treat them separately; we formalize this intuition in the
rest of this sub-section. Recall that from
Theorem~\ref{thm:intro-forest}, there is a $(1+\eps)$-approximate
solution (in the original graph $G$) that does not connect (that is,
pays the penalty for) terminal pairs in distinct components of $F$,
the forest returned by {\sc Primal-Dual Forest}. Let $F^*$ denote
this solution; $\cost(F^*) \le (1+\eps)\opt$. We construct a set of
prize-collecting Steiner forest instances, one for each tree $\T^j$ of
$F_2$. In instance $I^j$, we have $\pi^j_i = \pi_i$ for each pair
$(s_i, t_i)$ connected by $\T^j$, and $\pi^j_i = 0$ for all other
pairs. Let $\opt^j$ denote the cost of an optimal prize-collecting
Steiner forest to instance $I^j$; we prove the following theorem:

\begin{theorem}[Following \cite{BateniHM}]\label{thm:separateSpanners}
  $\sum_j \opt^j \le (1+\eps) \cost(F^*)$.
\end{theorem}

Given this theorem, we can separately solve each instance $I^j$; it is
easy to see that if we obtain a $\rho$-approximation to each instance,
combining them yields a solution of cost at most $\rho(1+\eps)
\cost(F^*) = \rho(1 + O(\eps))\opt$. But for each instance $I^j$, the
tree $\T^j$ contains \emph{all} terminal pairs with non-zero penalty,
and hence we can splice open the tree and convert the corresponding
Euler tour into the outer face to obtain a graph in which the length
of the outer face is bounded, and all relevant terminals are on this
outer face. This allows us to proceed with the spanner construction as
described in Appendix~\ref{app:spanner}.

Thus, it remains only to prove Theorem~\ref{thm:separateSpanners};
this closely follows the work of \cite{BateniHM}, with some additional
care needed because the cost of a forest includes both its length and
penalty.

\begin{proofof}{Theorem~\ref{thm:separateSpanners}}
  Recall that $F$ denotes the forest returned by {\sc Primal-Dual
    Forest}, and $F^*$ is an optimal prize-collecting Steiner forest
  that pays the penalty for terminal pairs in distinct components of
  $F$. Trees in $F$ were contracted to form vertices of $\hat{G}$; we
  then ran the prize-collecting clustering algorithm to form trees
  $\T^1, \T^2, \ldots$, making up the forest $F_2$.  We wish to show
  that if $\opt^j$ denotes the cost of an optimal solution for the
  instance $I^j$ (induced by terminal pairs in tree $\T^j$ of $F_2$),
  then $\sum_j \opt^j \le (1+\eps) \cost(F^*)$.

  To prove this theorem, we construct a set $\D$ of trees $\{T'_p\}$
  such that each $T'_p \in \D$ only connects terminal pairs in a
  single component of $F_2$. Further, $\sum_p \len(T'_p) \le (1+\eps)
  \len(F^*)$, and every pair connected by $F^*$ is connected by some
  $T'_p$. Such a set $\D$ of trees clearly yields solutions to the
  instances $I^j$, proving the theorem.

  We now construct the desired set $\D$. We begin by setting $\D$ to
  be the collection of trees in $F^*$, and then modify it as
  follows. Let $v_i$ be any vertex of $\hat{G}$ exhausted by $F^*$;
  prune from $F^*$ all terminals in the tree $T_i$ of $F$ contracted
  to form $v_i$, and add the tree $T_i$ to $\D$. When this process
  terminates, each tree in $\D$ only connects pairs in a single
  component of $F_2$.  Suppose this were not true; all trees
  \emph{added} to $\D$ clearly satisfy this condition, so it only
  remains to consider trees originally in $F^*$ (from which some
  terminals may have been pruned). But any tree $T^*$ of $F^*$
  connecting two vertices of $\hat{G}$ in distinct components of $F_2$
  must exhaust one of these vertices (from Lemma~\ref{lem:exhaust}),
  and hence the corresponding terminals should have been pruned from
  $T^*$, which yields a contradiction. It is also easy to see that any
  terminal pairs connected by $F^*$ are also connected by some tree in
  $\D$.

  To bound the cost of the trees in $\D$, we simply show that the
  length of the trees \emph{added} to $\D$ is at most $\eps \cdot
  \len(F^*)$. Let $v_i$ denote the vertex of $\hat{G}$ corresponding
  to the tree $T_i$ of $F$, and $V'$ the set of vertices exhausted by
  $F^*$. The length of the added trees is simply $\sum_{v_i \in V'}
  \len(T_i)$. From Lemma~\ref{lem:mergingComponents}, $\len(F^*) \ge
  \sum_{v_i \in V'} \phi_{v_i} = \sum_{v_i \in V'} \len(T_i)/\eps$.
  Rearranging, we get $\sum_{v_i \in V'} \len(T_i) \le \eps \cdot
  \len(F^*)$. 
\end{proofof}

\paragraph{Finishing the Spanner:} We can now complete the spanner
construction using the following theorem, implicit in the work of
\cite{BKM}:

\begin{theorem}\label{thm:spanner}
  Let $I$ be an instance of prize-collecting Steiner forest on a
  planar graph $G$. Let $F^*$ be an optimal solution to $I$. Given a
  tree $T$ spanning all terminals of $I$, for any fixed $\eps > 0$,
  there is a polynomial time algorithm to find a planar graph $H
  \subseteq G$ such that: (i) $\len(H) \le f(\eps) \cdot \len(T)$ for
  some function $f$ that depends (exponentially) on $\eps$ and (ii)
  there is a solution to instance $I$ in the graph $H$ of cost no more
  than $(1+\eps) \cost(F^*) + \eps \cdot \len(T)$.
\end{theorem}

The similar theorem stated in \cite{BKM} for instances of the Steiner
tree problem is slightly less general, though their proof technique
can be used to show the theorem we state here. For completeness, we
provide a proof in Appendix~\ref{app:spanner}.

\subsection{Completing the Reduction}\label{subsec:contraction}

In Section~\ref{subsec:scaled}, we constructed a forest $F_2$ by first
running the algorithm {\sc Primal-Dual Forest} on a modified instance
with scaled penalties, and then running the prize-collecting
clustering algorithm of \cite{BateniHM}. We proved two useful
properties of $F_2$: In Lemma~\ref{lem:pcCost}, we showed that
$\cost(F_2) \le (20/\eps^2) \opt$, and in
Theorem~\ref{thm:separateSpanners}, we argued that we could separately
solve a prize-collecting Steiner forest instance induced by each tree
$\T^j$ of $F_2$ without increasing the cost significantly.  (Formally,
we showed that $\sum_j \opt^j \le (1+\eps) \opt$.)

Now, for each tree $\T^j$ of $F_2$, construct a spanner for the
instance $I^j$. (Recall that $\T^j$ spans all terminals in $I^j$.)
Using a parameter $\eps' = \eps^3/20$, Theorem~\ref{thm:spanner}
guarantees a spanner $H^j$ for $I^j$ such that (i) $\len(H^j) =
f'(\eps) \cdot \len(\T^j)$ for some function $f'$ depending only on $\eps$,
and (ii) $\opt(H_j) \le (1+\eps') \opt^j + \eps'\cdot \len(\T^j)$, where
$\opt(H_j)$ denotes the cost of an optimal prize-collecting Steiner
forest in $H_j$. Hence,
\begin{equation}
  \sum_j \opt(H_j) \le (1+\eps') \sum_j \opt^j + \eps' \sum_j \len(\T^j)
   \le (1+2\eps) \opt + \frac{\eps^3}{20} \cdot \frac{20}{\eps^2} \opt
   = (1+3\eps) \opt
\end{equation}
where the second inequality follows from Lemma~\ref{lem:pcCost} and
Theorem~\ref{thm:separateSpanners}. 

Thus, if we can obtain a $\rho$-approximation to each instance $I^j$
in the graph $H^j$, we obtain a $\rho(1+3\eps)$-approximation to the
original prize-collecting Steiner forest instance. We use the
following theorem of \cite{DemaineHM}.

\begin{theorem}[Demaine, Hajiaghayi, Mohar \cite{DemaineHM}]
  \label{thm:tw-decomp}
  Let $G$ be any planar graph, and let $k$ be any integer such that $k
  \geq 2$. The edges of $G$ can be partitioned into $k$ sets such that
  contracting any one of the sets results in a graph of treewidth
  $O(k)$.  Furthermore, this partition can be found in polynomial
  time.
\end{theorem}

\begin{proofof}{Theorem~\ref{thm:reduction}}
  Let $H^j$ be the spanner for instance $I^j$ as constructed
  above. Set $k = (1 / \eps') \cdot f'(\eps)$, where $f'(\eps)$ is the
  function such that $\len(H^j) \le f'(\eps) \len(\T^j)$.

  Let $E_1, \cdots, E_k$ be the decomposition of the edges of $H^j$
  that is guaranteed by Theorem~\ref{thm:tw-decomp}. Let $E_{i^*}$ be
  the set of edges that has minimum length among the sets $E_1,
  \cdots, E_k$. We have
  $$\len(E_{i^*}) \leq {\len(H^j) \over k} 
    \leq {f'(\eps) \cdot \len(T^j) \over k} = \eps' \cdot \len(T^j).
  $$ 

  Let $\widehat{H^j} = H^j / E_{i^*}$; that is, $\widehat{H^j}$ is the
  graph obtained from $H^j$ by contracting the edges in $E_{i^*}$.  We
  assign new penalties to terminal pairs of $\widehat{H^j}$ in the
  natural way: When we contract an edge $uv$ into a single vertex $w$,
  we replace each terminal pair $(u,x)$ with a new pair $(w,x)$ with
  the same penalty as $(u,x)$, and we similarly replace each terminal
  pair $(v,y)$ with a new pair $(w,y)$ with the same penalty as
  $(v,y)$.  Let $\opt(\widehat{H^j})$ denote the cost of an optimal
  prize-collecting Steiner forest in $\widehat{H^j}$; it is obvious
  that $\opt(\widehat{H^j}) \leq \opt(H^k)$.

  Since $\widehat{H^j}$ has treewidth at most $k$, if there is a
  $\rho$-approximation for prize-collecting Steiner forest in graphs
  of fixed treewidth, we can find a $\rho$-approximate forest
  $\widehat{F^j}$ in $\widehat{H^j}$. We can then map $\widehat{F^j}$
  to a forest $F^j$ in $H^j$ using the edges in $E_{i^*}$. By
  construction,
  $$\cost(F^j) \leq \cost(\widehat{F^j}) + \len(E_{i^*}) \leq
  \rho \opt(H^j) + \eps' \cdot \len(\T^j).$$

  Combining such solution $F^j$ for each $H^j$, we find a forest of
  total cost $\sum_j \rho \opt(H^j) + \eps' \sum_j \len(\T^j)$. Using
  equation (1), the first term is at most $\rho (1+3\eps) \opt$, and
  from the choice of $\eps'$, the second term is at most $\eps \opt$. 
\end{proofof}

\section{Prize-Collecting Steiner Tree in Graphs of Fixed Treewidth}
\label{sec:treewidth}
In this section, for any fixed integer $k \ge 2$, we give a
polynomial-time algorithm to optimally solve the prize-collecting
Steiner tree problem in graphs of treewidth at most $k-1$.

A tree decomposition of a graph $G$ is a pair $(\script{T},
\script{B})$, where $\script{T} = (I, F)$ is a tree, and $\script{B} =
\{B_i \;|\; i \in I\}$ is a family of subsets of $V(G)$ such that
\begin{enumerate}[(1)]
\item $\bigcup_{i \in I} B_i = V(G)$
\item for every edge $uv \in E(G)$, there exists an
  $i$ such that $\{u, v\} \subseteq B_i$
\item for every vertex $v \in V(G)$, the set of
  nodes $\{i \in I \;|\; v \in B_i\}$ forms a
  connected subtree of $\script{T}$
\end{enumerate}
We refer to vertices of $\script{T}$ as nodes, and to each set of
$\script{B}$ as a bag.  The width of a tree decomposition
$(\script{T}, \script{B})$ is the size of the largest bag $B_i$ minus
one. As shown in \cite{Bodlaender}, for any fixed $k$, there is a
polynomial time algorithm (in fact a linear time algorithm) that
constructs a tree decomposition of $G$ of width at most $k$, or
reports that $G$ has treewidth greater than $k$. In the following, we
assume that we have a tree decomposition for $G$ of width at most $k -
1$, for some fixed $k$.

\medskip A tree decomposition $(\script{T}, \script{B})$ is {\em nice} if
the tree $\script{T}$ is rooted and, for every node $i \in I$, either
\begin{enumerate}[(1)]
\item $i$ has no children ($i$ is a leaf node)
\item $i$ has exactly two children $i_1, i_2$ and $B_{i_1} =
  B_{i_2} = B_i$ ($i$ is a join node)
\item $i$ has a single child $j$ and $B_i = B_j \cup \{v\}$
  for some vertex $v \in V(G)$ ($i$ is an introduce node)
\item $i$ has a single child $j$ and $B_i = B_j - \{v\}$ for
  some vertex $v \in V(G)$ ($i$ is a forget node)
\end{enumerate}

The following lemma is well-known and it is straightforward to prove.
\begin{lemma}
  There is a linear time algorithm that, given a tree decomposition
  for $G$, constructs a nice tree decomposition $(\script{T},
  \script{B})$ of the same width. Moreover, the tree $\script{T}$ has
  $O(|V|)$ nodes.
\end{lemma}

\subsection{A Dynamic Program for Prize-Collecting Steiner
  Tree}\label{subsec:pcstTreewidth}

We solve the problem using dynamic programming on a nice tree
decomposition $(\script{T}, \script{B})$ of width $k - 1$. For each
node $i \in I$, let $V_i$ be the set of all vertices appearing in the
bags corresponding to the nodes of the subtree of $\script{T}$ rooted
at $i$. Let $G_i$ be the subgraph of $G$ induced by $V_i$. We start
with an informal overview of the algorithm. For the purposes of
exposition, we assume that there is only one optimal prize-collecting
Steiner tree $T^*$ (if there are several solutions, we fix one of
them).  Additionally, we assume without loss of generality that the
root vertex $r$ is in the bag corresponding to the root node of
$\script{T}$. Now fix a node $i$ of $\script{T}$, and consider the
graph $G_i$. Clearly, we would like to compute the subgraph $F$ of
$T^*$ that lies in $G_i$. In order to do so, we will specify some
information about this subgraph $F$. More precisely, we will specify
the subgraph $H$ of $T^*$ that lies in $G[B_i]$, and a partition
$\alpha$ of the vertices in $H$ induced by the connected components of
$F$, i.e., each part of $\alpha$ consists of all the vertices in $H$
that are in the same connected component of $F$.  (Intuitively,
$\alpha$ tells us that we need to connect each part using a tree of
$G_i$ and all of these trees are guaranteed to be connected to the
root outside $G_i$.) It follows from the optimality of $T^*$ that $F$
is a minimum cost subgraph of $G_i$ satisfying
	\begin{enumerate}[$(c_1)$]
		\item $F[B_i] = H$
		\item the partition of $V(H)$ induced by the connected
		components of $H$ is a refinement of $\alpha$ (if two vertices
		$u, v$ are in the same connected component of $H$, then $u$
		and $v$ are in the same part of $\alpha$)
		\item the partition of $V(H)$ induced by the connected
		components of $F$ is $\alpha$
	\end{enumerate}
Let $c(i, H, \alpha)$ be the minimum cost of a subgraph $F$ of $G_i$
satisfying $(c_1) - (c_3)$. We will compute $c(i, H, \alpha)$ for all
valid tuples $(i, H, \alpha)$ using dynamic programming. The cost of
the optimal prize-collecting Steiner tree is equal to $\min_{H,
\alpha} \; c(r', H, \alpha)$, where $r'$ is the root node of
$\script{T}$, and the minimum is over all pairs $(H, \alpha)$ such
that $H$ is a subgraph of $G[B_{r'}]$ containing $r$, and $\alpha$ has
a single part containing the vertices of $H$. To see why this is true,
consider a pair $(H, \alpha)$ such that $H$ is a subgraph of
$G[B_{r'}]$ containing $r$, and $\alpha$ consists of a single part
containing the vertices of $H$. If $T$ is a solution for the
subproblem $(r', H, \alpha)$ then $T$ is a tree that contains $r$, and
hence $T$ is a valid prize-collecting Steiner tree. Conversely, let
$T$ be a prize-collecting Steiner tree. Let $H = T[B_{r'}]$ and let
$\alpha$ be the partition of $V(H)$ induced by $T$. Since $T$ is a
tree containing $r$, $H$ contains $r$ and $\alpha$ consists of a
single part containing the vertices of $H$. Therefore $T$ is a valid
solution for the subproblem $(r', H, \alpha)$.

Let $i$ be a node of $\script{T}$. Then $i$ is a leaf node, a join
node, an introduce node, or a forget node, and we consider each of
these cases separately. Before we describe the recurrence for $c(i, H,
\alpha)$, we introduce some useful terminology (borrowed from
\cite{BateniHM}).

We can view a partition $\alpha$ as an equivalence relation over the
vertices, and we write $u \equiv_{\alpha} v$ if $u$ and $v$ are in the
same part of $\alpha$. Let $\alpha_1$ and $\alpha_2$ be two partitions
of the same vertex set. We say that $\alpha_1$ is \emph{finer} than
$\alpha_2$ --- or equivalently, that $\alpha_1$ is a refinement of
$\alpha_2$ --- if $u \equiv_{\alpha_1} v$ implies $u \equiv_{\alpha_2}
v$. If $\alpha_1$ is finer than $\alpha_2$, we say that $\alpha_2$ is
\emph{coarser} than $\alpha_1$. We use $\alpha_1 \vee \alpha_2$ to
denote the finest partition that is coarser than both $\alpha_1$ and
$\alpha_2$ (there is a unique such partition).

~\\
\textbf{Node $i$ is a leaf node.} Let $\beta$ be the partition of
$V(H)$ induced by the connected components of $H$. We have
	\begin{equation}\label{btw:eq-leaf}
		c(i, H, \alpha) =
			\begin{cases}
				\len(H) + \pen(B_i - V(H)) \qquad \text{if $\alpha =
				\beta$}\\
				\infty \qquad \text{otherwise}
			\end{cases}
	\end{equation}

\begin{proofof}{Equation~\ref{btw:eq-leaf}}
	Since $G_i = G[B_i]$, $H$ is the only subgraph satisfying $(c_1)$.
	If $\alpha \neq \beta$, there is no subgraph satisfying $(c_1) -
	(c_3)$. Otherwise, $H$ is the only subgraph satisfying $(c_1) -
	(c_3)$ and its cost is $\len(H) + \pen(H)$.
\end{proofof}

~\\
\textbf{Node $i$ is a join node.} Let $i_1$ and $i_2$ be the
children of $i$. We have
	\begin{equation}\label{btw:eq-join}
		c(i, H, \alpha) = \min_{\alpha_1, \alpha_2} \left(c(i_1, H,
		\alpha_1) + c(i_2, H, \alpha_2) - \len(H)\right)
	\end{equation}
where the minimum is taken over all partitions $\alpha_1$, $\alpha_2$
of $V(H)$ such that $\alpha = \alpha_1 \vee \alpha_2$.

~\\
The intuition behind Equation~\ref{btw:eq-join} is the following. Let
$F$, $F_1$, $F_2$ denote the restrictions of the optimal tree $T^*$ to
$G_i$, $G_{i_1}$, $G_{i_2}$ (respectively). Then $F$ is the union of
$F_1$ and $F_2$. Let $\alpha$, $\alpha_1$, $\alpha_2$ be the
partitions of $B_i \cap V(T^*)$ induced by the connected components of
$F, F_1, F_2$. Since $F_1$ and $F_2$ intersect only at $B_i \cap
V(T^*)$, $\alpha = \alpha_1 \vee \alpha_2$.

Formal proofs of correctness for this, and subsequent cases, can be found in
Appendix~\ref{app:dp-proofs}.

~\\
\textbf{Node $i$ is a forget node.} Let $j$ be the child of $i$,
and let $v$ be the vertex in $B_j - B_i$. Fix a subgraph $H$ of
$G[B_i]$, and a partition $\alpha$ of $V(H)$. Let $S$ be a subset of
the neighbors of $v$ that are in $B_i$.  Let $E(v, S)$ denote the
edges with an endpoint in $v$ and the other in $S$. Let $\alpha(v, S)$
be the partition of $V(H) \cup \{v\}$ obtained from $\alpha$ as
follows: we merge each part of $\alpha$ that contains a vertex in $S$
into a single part and add $v$ to it; we add all remaining parts of
$\alpha$ to $\alpha(v, S)$. We have

\begin{equation}\label{btw:eq-forget}
		c(i, H, \alpha) = \min\left(c(j, H, \alpha), \min_{S \subseteq
		V(H) \cap \Gamma(v)} \left(c(j, H \cup \{v\} \cup E(v, S),
		\alpha(v, S)\right)\right)
	\end{equation}
where the second minimum is taken over all sets $S \subseteq V(H) \cap
\Gamma(v)$ such that $S$ has at most one vertex in each part of
$\alpha$.

~\\
The intuition behind Equation~\ref{btw:eq-forget} is the following.
Let $F, F'$ denote the restriction of the optimal tree $T^*$ to $G_i,
G_j$ (respectively). If $T^*$ does not contain $v$, we have $F' = F$.
Therefore we may assume that $T^*$ contains $v$, and thus $F'$
consists of $F$ and the edges of $E(T^*) \cap E(G_j)$ that are
incident to $v$. The edges of $F'$ that are incident to $v$ have at
most one endpoint in each connected component of $F$. Thus each
connected component of $F'$ that does not contain $v$ is a connected
component of $F$, and the connected component of $F'$ containing $v$
consists of one or more connected components of $F$ that connect to
each other via the edges incident to $v$.

~\\
\textbf{Node $i$ is an introduce node.} Let $j$ be the child of $i$, and
let $v$ be the vertex in $B_i - B_j$. Let $S$ be the set of all
neighbors $u$ of $v$ such that the edge $uv$ is in $H$. For each
partition $\alpha'$ of $V(H) - v$, we let $\alpha'(v, S)$ be the
partition of $V(H)$ obtained from $\alpha'$ as follows: we merge each
part of $\alpha'$ that contains a vertex in $S$ into a single part and
add $v$ to it; we add all remaining parts of $\alpha'$ to $\alpha'(v,
S)$. We have
	\begin{equation}\label{btw:eq-introduce}
		c(i, H, \alpha) =
			\begin{cases}
				c(j, H, \alpha) + \pen(v)
				\qquad\qquad\qquad\qquad\qquad\qquad\qquad \text{if $v
				\notin V(H)$}\\
				\min_{\alpha'} \left(c(j, H - v, \alpha') + \sum_{uv
				\in H} \len(uv)\right) \qquad \;\; \text{otherwise}
			\end{cases}
	\end{equation}
where the minimum is taken over all partitions $\alpha'$ of
$V(H) - \{v\}$ satisfying
	\begin{enumerate}[$(i)$]
		\item $S$ has at most one vertex in each part of
		$\alpha'$
		\item $\alpha'(v, S) = \alpha$
	\end{enumerate}
(Note that there exists a partition $\alpha'$ that satisfies
the conditions above.)

~\\
The intuition behind Equation~\ref{btw:eq-introduce} is the following.
Let $F, F'$ denote the restriction of the optimal tree $T^*$ to $G_i,
G_j$ (respectively). If $T^*$ does not contain $v$, we have $F = F'$.
Therefore we may assume that $T^*$ contains $v$, and thus $F$ consists
of $F'$ together with the edges of $E(T^*) \cap E(G[B_i])$ that are
incident to $v$.  The edges of $F$ that are incident to $v$ have at
most one endpoint in each connected component of $F'$. Thus each
connected component of $F$ that does not contain $v$ is a connected
component of $F'$, and the connected component of $F$ containing $v$
consists of one or more connected components of $F'$ that connect to
each other via the edges incident to $v$.

\begin{proofof}{Theorem~\ref{thm:pcstBoundedTreewidth}}
  Let $b_k$ be the number of partitions of a $k$-element
  set\footnote{The Bell number $B_{k}$ is the number of partitions of
    a $k$-element set. To avoid confusion with the bags of the tree
    decomposition, we will use $b_k$ to refer to the $k$-th Bell
    number.}, and let $s_k$ be the number of subgraphs of a graph with
  $k$ vertices. Since each bag has at most $k$ vertices and
  $\script{T}$ has $O(|V|)$ nodes, there are $O(|V| \cdot b_k \cdot
  s_k)$ distinct subproblems.  Additionally, we can evaluate each
  subproblem in $O(b^2_k)$ time once we have a solution for each of
  the subproblems it depends on. (The most expensive evaluation
  corresponds to a join node.) Therefore we can find an optimal
  prize-collecting Steiner tree in $O(b^3_k \cdot s_k \cdot |V|)$
  time.
\end{proofof}

\paragraph{Acknowledgments:} The work on this paper evolved during a
reading group on this topic. We thank Kyle Fox, Sariel Har-Peled,
Nirman Kumar, Amir Nayyeri and Ben Raichel for various discussions,
and in particular Sariel for his ideas and intuition on this and
related topics. We thank MohammadHossein Bateni, MohammadTaghi
Hajiaghayi and Daniel Marx for sharing their results and manuscript
\cite{BateniHM-pcst} with us, and for pointing out related work.

\bibliographystyle{plain}
\bibliography{pcst-planar}

\appendix

\section{Constructing the Spanner: A Proof of
  Theorem~\ref{thm:spanner}}
\label{app:spanner} 
\noindent
We begin by duplicating the edges of $T$, and introducing multiple
copies of its non-leaf vertices in order to transform the Euler tour
corresponding to $T$ into a cycle. Let $G'$ by the resulting graph.
We then make this cycle the outer face $\Delta$ of $G'$.

\begin{definition}[Definition~{6.2}, Borradaile \etal \cite{BKM}]
  A path $P$ is $\eps$-short in $G'$ if for every pair of vertices
  $x$ and $y$ on $P$, the distance from $x$ to $y$ along $P$ is at
  most $(1 + \eps)$ times the distance from $x$ to $y$ in $G'$
  (i.e., $dist_P(x, y) \leq (1 + \eps) dist_{G'}(x, y)$).
\end{definition}

\textbf{Strips}.  Let $\Delta[x, y]$ denote the subpath of the outer
face $\Delta$ from $x$ to $y$. We find a pair of vertices $x, y$ on
$\Delta$ such that $\Delta[x, y]$ is a minimal subpath of $\Delta$
that is not $\eps$-short in $G'$. Let $N$ be a shortest path from $x$
to $y$ in $G'$. The subgraph enclosed by $\Delta[x, y] \cup N$ is a
{\em strip}. We recursively decompose the subgraph of $G'$ enclosed by
$N \cup (\Delta - \Delta[x, y])$ into strips, if the graph is
nontrivial.

\begin{lemma}[Lemma~{6.3}, Borradaile \etal \cite{BKM}]
  \label{lem:strips}
  The total length of all the boundary edges of all the strips is at
  most $(\eps^{-1} + 1) \cdot \len(\Delta)$.
\end{lemma}

\noindent
\textbf{Columns}.  Consider a strip, with north and south boundaries
$N$ and $S$ ($N$ is the shortest path we added when we created the
strip). We select vertices $s_0, s_1, \dots$ on $S$ and paths $C_0,
C_1, \dots$ inside the strip as follows. The vertex $s_0$ is the left
endpoint common to $S$ and $N$, and column $C_0$ is the (empty)
shortest path from $s_0$ to $N$.  Now suppose that we have selected
vertices $s_0, s_1, \dots, s_{i - 1}$ and columns $C_0, C_1, \dots,
C_{i - 1}$. The vertex $s_i$ is the first vertex on $S$ such that the
distance from $s_{i - 1}$ to $s_i$ on $S$ is greater than $\eps$ times
the distance from $s_i$ to $N$ in the strip, and the column $C_i$ is
the shortest path in the strip from $s_i$ to $N$.

\begin{lemma}[Lemma~{6.4}, Borradaile \etal \cite{BKM}]
  \label{lem:columns}
  The sum of the lengths of the columns in a strip is at most $\len(S)
  / \eps$, where $S$ is the south boundary of the strip.
\end{lemma}

\noindent
\textbf{Supercolumns}.  Let
\[k = {1 \over \eps^2} \left( {1 \over \eps} + 1 \right)\] For each
strip, we select a subset of the columns $\{C_0, C_1, \dots\}$ of the
strip as follows. Let
$$\script{C}_i = \{C_j \; | \; j \equiv i \mod k\}$$
where $0 \leq i \leq k - 1$. Let $i^*$ be the index that minimizes
$\len(\script{C}_i)$. We designate the columns in $\script{C}_{i^*}$
as the supercolumns of the strip.

\begin{lemma}[Lemma~{6.5}, Borradaile \etal \cite{BKM}]
  \label{lem:supercolumns}
  The sum of the lengths of the supercolumns in a strip is at most $1
  / k$ times the sum of the lengths of the columns in the strip.
\end{lemma}

\begin{lemma} \label{lem:supercolumns2} 
  The sum over all strips of the length of all the supercolumns is at
  most $\eps \cdot \len(\Delta)$.
\end{lemma}
\begin{proof}
  By Lemma~\ref{lem:strips} and Lemma~\ref{lem:columns},
  the total length of the columns is at most
  \[{1 \over \eps} \left( {1 \over \eps} + 1 \right) \len(\Delta)\]
  By Lemma~\ref{lem:supercolumns}, the total length of the
  supercolumns is at most 
  $${1 \over k} \cdot \left( \text{total length of columns} \right)
  \leq  {1 \over k} \cdot {1 \over \eps} 
     \left( {1 \over \eps} + 1 \right) \len(\Delta)
   = \eps \cdot \len(\Delta)$$

\vspace{-11mm}
\end{proof}

\bigskip \noindent
\textbf{Mortar Graph}. The mortar graph $MG$ is a subgraph of
\emph{the original graph} $G$ consisting of the edges of the given
tree $T$ (that was doubled to form the outer face $\Delta$ of $G'$),
the edges of the shortest paths that define the strips, and the edges
of the supercolumns.

\begin{lemma} \label{lem:MG} 
  The length of the mortar graph $MG$ is at most $({3 \over \eps} +
  \eps) \cdot \len(\Delta)$.
\end{lemma}
\begin{proof}
  The total length of the strips is at most
  $\left( {1 \over \eps} + 1\right) \len(\Delta)$.
  The total length of the supercolumns is at most $\eps \cdot
  \len(\Delta)$. The length of $T$ is precisely half the length of
  $\Delta$, which consisted of two copies of each edge of $T$.  Thus,
  the total length of $MG$ is at most $({1 \over \eps} + 1.5 + \eps
  )\cdot \len(\Delta)$
\end{proof}

\medskip
\begin{proposition} \label{prop:MG2}
  The mortar graph $MG$ contains every vertex of $T$.
\end{proposition}

\noindent
\textbf{Bricks}.  A brick consists of all edges of the original graph
$G$ that are (strictly) enclosed by the boundary of some face $f$ of
the mortar graph. (Note that if an edge $e$ on the outer face of $G$
is not a part of $MG$, it is ``enclosed'' by the outer face of $MG$.)
For each face $f$ of the mortar graph that encloses at least one edge,
there is a corresponding brick.

\begin{lemma}[Lemma~{6.10}, Borradaile \etal \cite{BKM}]
  \label{lem:bricks}
  The boundary $\partial B$ of a brick $B$, in counterclockwise order,
  is the concatenation of four paths $W_B$, $S_B$, $E_B$, $N_B$ such
  that
  \begin{enumerate}
  \item The set of edges of $B - \partial B$ is
    non-empty.
  \item Every vertex of $T$ that is in $B$ is on $S_B$
    or $N_B$.
  \item $N_B$ is $0$-short in $B$, and every proper
    subpath of $S_B$ is $\eps$-short in $B$.
  \item There exists a number $k' \leq k$ and vertices
    $s_0, s_1, \dots, s_{k'}$ ordered west to east on
    $S_B$ such that, for each $i$ and each vertex $x$ on
    $S_B[s_i, s_{i + 1})$, $dist_{S_B}(x, s_i) <
    \eps \cdot dist_B(x, N_B)$.
  \end{enumerate}
\end{lemma}

\begin{definition}[Definition~{10.3}, Borradaile \etal \cite{BKM}] 
  Let $H$ be a subgraph of $G$ such that $P$ is a path in $H$. A {\em
    joining vertex} of $H$ with $P$ is a vertex of $P$ that is the
  endpoint of an edge of $H - P$.
\end{definition}

\begin{lemma}[Theorem~{10.7}, Borradaile \etal \cite{BKM}]
  \label{lem:bricks2}
  Let $B$ be a plane graph with boundary $W \cup S \cup E \cup N$,
  satisfying the brick properties of {\em Lemma~\ref{lem:bricks}}. Let
  $F$ be a set of edges of $B$. There is a forest $\tilde{F}$ of $B$
  with the following properties:
  \begin{enumerate}
  \item If two vertices of $N \cup S$ are connected in
    $F$ then they are connected in $\tilde{F}$.
  \item The number of joining vertices of $F$ with
    both $N$ and $S$ is at most $\alpha(\eps)$,
    where $\alpha(\eps) = o(\eps^{-5.5})$.
  \item $\len(\tilde{F}) \leq (1 + c \eps)
    \len(F)$, for some fixed constant $c$.
  \end{enumerate}
\end{lemma}

\noindent
\textbf{Portals}.  Let $\theta = \theta(\eps)$ be a parameter that
depends polynomially on $1 / \eps$. For each brick $B$, we designate
some vertices of $\partial B$ as portals, evenly spaced around $B$ as
follows. Let $v_0 \in \partial B$ be the endpoint of an edge strictly
enclosed by $\partial B$; we designate $v_0$ as a portal. Now suppose
we have designated $v_0, v_1, \dots, v_{i - 1}$ as portals. Let $v_i$
be the first vertex on $\partial B$ such that $\len(\partial B[v_{i -
  1}, v_i]) > \len(\partial B) / \theta$. We designate $v_i$ as a
portal, unless $v_0 \in V(\partial B(v_{i - 1}, v_i])$, in which case
we stop.

\begin{lemma}[Lemma~{7.1}, Borradaile \etal \cite{BKM}]
  \label{lem:portalCoverage} 
  For any vertex $x$ on $\partial B$, there is a portal $y$ such that
  the $x$-to-$y$ subpath of $\partial B$ has length at most
  $\len(\partial B) / \theta$.
\end{lemma}

\begin{lemma}[Lemma~{7.2}, Borradaile \etal \cite{BKM}]
  \label{lem:portalCardinality} There are at most $\theta$
  portals on $\partial B$.
\end{lemma}

\noindent
\textbf{Portal-connected graph}.  For any subgraph $G''$ of the mortar
graph $MG$, we construct a planar graph $\script{B}^+(G'')$ as
follows. For each face $f$ of $G'$ corresponding to a brick $B$, we
embed a copy of $B$ inside the face $f$, and, for each portal $v$ of
$B$, we connect the copy of $v$ in the brick with the copy of $v$ on
$f$ using a zero-length edge. We refer to these zero-length edges as
{\em portal edges}, and we refer to $\script{B}^+(MG)$ as the {\em
  portal-connected graph}.  Finally, any new vertex receives penalty
zero (these vertices are copies of vertices of the mortar graph).

\begin{theorem} \label{thm:structureThm} 
  Let $F^*$ be an optimal Prize-Collecting Steiner forest in $G$. There
  exists a constant $\theta = \theta(\eps)$ depending polynomially on
  $1 / \eps$ such that, for any choice of portals satisfying the
  condition in Lemma~\ref{lem:portalCoverage}, the corresponding
  portal-connected graph $\script{B}^+(MG)$ contains a forest
  $\widehat{F}$ with the following properties:
  \begin{enumerate}
  \item $\len(\widehat{F}) \leq (1 + c_1\eps)\len(F^*) + c_2\eps \cdot
    \len(\Delta)$, where $c_1, c_2$ are absolute constants
  \item Any two vertices of $MG$ in the same component of $F^*$ are
    connected by $\widehat{F}$.
  \end{enumerate}
\end{theorem}
\begin{proof}
  Let $F^*$ be an optimal tree to the Prize-Collecting Steiner forest
  problem in $G$, and let $S^*$ be the set of all vertices of the
  mortar graph $MG$ that are in $F^*$.  We will follow the proof of
  Theorem~{3.2} in Borradaile \etal \cite{BKM}; there are two main
  steps. First, we transform $F^*$ into a solution in $G$ that only
  has a few joining vertices in each brick. To convert this into the
  desired forest $\widehat{F}$ in the portal-connected graph
  $\script{B}^+(MG)$, we simply add edges connecting the joining
  vertices in each brick to the nearest portals. As there are not many
  joining vertices, we can connect them to the portals without
  significantly increasing the cost. We describe the process
  completely below.

  \bigskip
  First, we add the east and west boundaries of each brick; let $F_1$
  be the union of $F^*$ with the east and west boundaries ($E_B$ and
  $W_B$) for each brick $B$. By Lemma~\ref{lem:supercolumns2},
  $\len(F_1) \leq \len(F^*) + \eps \cdot \len(\Delta)$.
  Clearly, $F_1$ connects all vertices of $S^*$ connected by $F^*$.

  Next, we reduce the number of joining vertices on the north and
  south boundaries of each brick. Let $F_{1|B}$ be the subgraph of $F_1$
  that is strictly embedded in a brick $B$ of $G$. We replace $F_{1|B}$ with
  the forest $F_{2|B}$ that is guaranteed by Lemma~\ref{lem:bricks2}. We
  have
  $$\len(F_{2|B}) \leq (1 + c_1\eps) \len(F_{1|B})$$
  Let $N$ and $S$ denote the north and south boundaries of the
  brick. Since any two vertices of $N \cup S$ that are connected in
  $F_{1|B}$ are also connected in $F_{2|B}$, it follows that $F_{2|B}$
  connects all vertices of $S^*$ connected by $F_{1|B}$.

  We apply this procedure for each brick in order to get a subgraph
  $F_2$. Since the bricks are disjoint,
  \begin{eqnarray*}
    \len(F_2) &\leq& (1 + c_1\eps) \len(F_1)\\
    &\leq& (1 + c_1 \eps) \left(\len(F^*) + \eps \cdot \len(\Delta)\right)\\
    &=& (1+c_1\eps) \len(F^*) + (\eps + c_1 \eps^2) \len(\Delta)\\
  \end{eqnarray*}

  \vspace{-8mm}
  Moreover, $F_2$ connects all vertices of $S^*$ connected by $F_1$.

  \bigskip Now we convert the forest $F_2 \subseteq G$ to a subgraph
  of $\script{B}^+(MG)$.  Note that every edge of $G$ has at least one
  corresponding edge in $\script{B}^+(MG)$ (an edge $e$ of $MG$ has
  three copies: one mortar edge, and one inside each of the bricks
  corresponding to the two faces of $G'$ incident to $e$). For each
  edge $e$ of $F_2$, we select a corresponding edge of
  $\script{B}^+(MG)$ as follows. If $e$ is an edge of $MG$, we select
  the corresponding mortar edge of $\script{B}^+(MG)$.
  Otherwise, we select the unique edge corresponding to $e$ in
  $\script{B}^+(MG)$. Let $F_3$ denote the resulting subgraph of
  $\script{B}^+(MG)$. We have:
  $$\len(F_3) = \len(F_2) 
  \leq (1+c_1\eps) \len(F^*) + (\eps + c_1 \eps^2) \len(\Delta).$$
  
  Since $F_3$ does not connect the connected components of $F_2$, we
  connect it using portal edges and mortar edges as follows. Consider
  a brick $B$, and let $V_B$ denote the set of joining vertices of
  $F_3$ with $N_B \cup S_B$. For each vertex $v \in V_B$, let $p_v$ be
  the portal vertex that is closest to $v$, let $P_v$ be the shortest
  $v$-to-$p_v$ path along $\partial B$, and let $P'_v$ be the
  corresponding path of mortar edges. Let $e_v$ be the portal edge
  corresponding to $p_v$. We add $P_v$, $P'_v$, and $e_v$ to $F_3$. We
  apply this procedure for each brick in order to get the subgraph
  $\widehat{F}$. First, we bound the length of $\widehat{F}$.
  \begin{eqnarray*}
    \len(\widehat{F}) &\leq& \len(F_3) + 
      \sum_{B} \sum_{v \in V_B} (\len(P_v) + \len(e_v) + \len(P'_v))\\
    &=& \len(F_3) + 2 \sum_{B} \sum_{v \in V_B} \len(P_v) 
      \qquad\qquad \mbox{[$\len(e) = 0$, $\len(P'_v) = \len(P_v)$]}\\
    &\leq& \len(F_3) + 
      2 \sum_{B} \sum_{v \in V_B} \len(\partial B) / \theta(\eps) 
      \qquad \qquad \mbox{[Lemma~\ref{lem:portalCoverage}]}\\
    &\leq& \len(F_3) + 2 \sum_{B} \alpha(\eps) \len(\partial B) /
      \theta(\eps) \qquad \qquad \mbox{[Lemma~\ref{lem:bricks2}]}\\
    &=& \len(F_3) + {2 \alpha(\eps) \over \theta(\eps)} 
      \cdot \sum_{B} \len(\delta(B))\\
    &\leq& \len(F_3) + {4 \alpha(\eps) \over \theta(\eps)} \cdot \len(MG)\\
    &\leq& \len(F_3) + {16 \alpha(\eps) \over \eps \theta(\eps)} 
      \cdot \len(\Delta) \qquad \qquad \mbox{[Lemma~\ref{lem:MG}]}
  \end{eqnarray*}
  Setting $\theta(\eps) = 16 \eps^{-2} \alpha(\eps)$ gives us
  \begin{eqnarray*}
    \len(\widehat{F}) &\leq& \len(F_3) + \eps \cdot \len(\Delta)\\
    &\leq& (1+c_1\eps)\len(F^*) + (2 + c_1\eps) \eps \cdot \len(\Delta)\\
  \end{eqnarray*}

  \vspace{-6mm}
  It remains only to show that $\widehat{F}$ connects any two vertices
  of $S^*$ connected by $F_2$, and hence by $F^*$.  Let $x$ and $y$ be
  two vertices of $S^*$ connected by $F_2$ via an $x - y$ path $P$ in
  $F_2$. The definition of $F_3$ breaks $P$ into disjoint paths.
  Consider one such path $P_i$ that is not a subpath of $MG$. By
  construction, the endpoints of $P_i$ are joining vertices. When we
  construct $\widehat{F}$ from $F_3$, we connect the endpoints of
  $P_i$ to their corresponding vertices on $MG$ via portal edges.
  Therefore there is an $x-y$ path in $\widehat{F}$.
\end{proof}

\bigskip \noindent \textbf{Spanner}.  For each brick $B$ and for each
subset $X$ of the portals of $B$, we find an optimal Steiner Tree for
$B$ and $X$. The {\em spanner} $H$ consists of all edges of these
Steiner Trees together with the edges of the mortar graph $MG$.

\begin{lemma} \label{lem:spanner1}
  The total length of the spanner $H$ is at most $(1 + 2^{1 +
    \theta}) (\frac{3}{\eps} + \eps) \len(\Delta)$.
\end{lemma}
\begin{proof}
  As shown in Lemma~{4.1} of \cite{BKM}, the total length of all
  Steiner trees is at most $2^{1 + \theta} \cdot \len(MG)$. Thus the
  length of $H$ is at most $(1 + 2^{1 + \theta}) \cdot
  \len(MG)$. Lemma~\ref{lem:MG} completes the proof.
\end{proof}

\begin{lemma} \label{lem:spanner2} The spanner $H$ contains a
  prize-collecting Steiner forest $F'$ of cost at most $(1 + c_1\eps)
  \len(F^*) + c_2 \eps \cdot \len(\Delta) + \pen(F^*)$, for some
  absolute constants $c_1, c_2$.
\end{lemma}
\begin{proof}
  We will follow the proof of Lemma~{4.2} in Borradaile \etal
  \cite{BKM}. Let $F^*$ be an optimal forest in $G$ and let
  $\widehat{F}$ be the forest guaranteed by
  Theorem~\ref{thm:structureThm}.  For each brick $B$ and for each
  connected component $K$ of the intersection of $\widehat{F}$ with
  $B$, let $X$ be the set of portals of $B$ belonging to $K$; we
  replace $K$ with the optimal Steiner Tree for $B$ and $X$ contained
  in the spanner.  Let $\tilde{F}$ be the subgraph resulting from all
  these replacements. We have
  \[\len(\tilde{F}) \leq \len(\widehat{F}) \leq
  (1 + c_1 \eps) \len(F^*) + c_2\eps \cdot \len(\Delta)\]

  Moreover, since $\widehat{F}$ connects all vertices of $MG$
  connected by $F^*$, and all terminals are vertices of $MG$ as they
  are connected by $T$, it follows that $\tilde{F}$ also connects all
  terminals connected by $F^*$. Hence, $\pen(\tilde{F}) \le
  \pen(F^*)$; as $\cost(\tilde{F}) = \len(\tilde{F}) +
  \pen(\tilde{F})$, we obtain the lemma.
\end{proof}

Theorem~\ref{thm:spanner} now follows almost directly from
Lemmas~\ref{lem:spanner1} and \ref{lem:spanner2}, as $\len(\Delta) = 2
\cdot \len(T)$; simply construct the spanner $H$ using a modified parameter
$\eps' = {\eps \over \max\{c_1, 2c_2\}}$.

\section{Omitted Proofs from Section~\ref{subsec:pcstTreewidth}}
\label{app:dp-proofs}
\begin{proofof}{Equation~\ref{btw:eq-join}}
  Let $(\alpha_1, \alpha_2)$ be a pair that minimizes the right hand
  side of Equation~\ref{btw:eq-join}.  Let $F_{\ell}$ be an optimal
  solution for the subproblem $(i_{\ell}, H, \alpha_{\ell})$, where
  $\ell = 1, 2$. Let $F = F_1 \cup F_2$. Now we claim that $F$ is a
  solution for the subproblem $(i, H, \alpha)$, i.e., it satisfies the
  conditions $(c_1) - (c_3)$.  Clearly, $F$ satisfies $(c_1)$ and
  $(c_2)$. Now let $u$ and $v$ be two vertices in the same part of
  $\alpha$. Since $\alpha = \alpha_1 \vee \alpha_2$, $u$ and $v$ are
  in the same part of $\alpha_{\ell}$ for some $\ell$, and thus $u$
  and $v$ are in the same connected component of $F_{\ell}$.  Thus the
  partition of $V(H)$ induced by $F$ is coarser than $\alpha$. Since
  $\alpha$ is coarser than $\alpha_1$ and $\alpha_2$, it follows that
  $\alpha$ is coarser than the partition of $V(H)$ induced by $F$ as
  well.  Therefore $F$ satisfies $(c_3)$. Since $E(H) \subseteq E(F_1)
  \cap E(F_2)$,
  $$\len(F) \leq \len(F_1) + \len(F_2) - \len(H)$$
  Since $V(F) = V(F_1) \cup V(F_2)$ and $V(G_{i}) = V(G_{i_1}) \cup
  V(G_{i_2})$, we have
  $$\pen(V(G_i) - V(F)) \leq \pen(V(G_{i_1}) - V(F_1)) +
  \pen(V(G_{i_2}) - V(F_2))$$
  Thus
  $$c(i, H, \alpha) \leq \cost(F) \leq \cost(F_1) + \cost(F_2) -
  \len(H) = c(i_1, H, \alpha_1) + c(i_2, H, \alpha_2) - \len(H)$$
  Conversely, let $F$ be an optimal solution for the subproblem $(i,
  H, \alpha)$. Let $F_{\ell}$ be the restriction of $F$ to
  $G_{i_{\ell}}$, where $\ell = 1, 2$. Let $\alpha_{\ell}$ be the
  partition induced by the connected components of $F_{\ell}$. Now we
  claim that $F_{\ell}$ is a solution for the subproblem $(i_{\ell},
  H, \alpha_{\ell})$: by construction, $F_{\ell}$ satisfies $(c_1) -
  (c_3)$. Since $F_1$ and $F_2$ intersect only at $V(H)$, $\alpha =
  \alpha_1 \vee \alpha_2$. Therefore the right hand side of the
  equation is at most
  $$\cost(F_1) + \cost(F_2)  - \len(H) \leq \cost(F) = c(i, H,
  \alpha)$$
  which completes the proof.
\end{proofof}

~\\
\begin{proofof}{Equation~\ref{btw:eq-forget}}
  Suppose the minimum of the right hand side of the equation is
  achieved by an optimal solution $F'$ for the subproblem $(j, H,
  \alpha)$. Since $v$ is not in $H$, $F'$ does not contain $v$. Thus
  $F'$ is a solution for the subproblem $(i, H, \alpha)$, and
  therefore $c(i, H, \alpha)$ is at most the right hand side of
  Equation~\ref{btw:eq-forget}.  Therefore we may assume that the
  minimum of the right hand side is achieved by an optimal solution
  $F'_S$ for the subproblem $(j, H \cup \{v\} \cup E(v, S), \alpha(v,
  S))$. Let $F = F'_S - E(v, S) - \{v\}$.  Now we claim that $F$ is a
  solution for the subproblem $(i, H, \alpha)$. By construction, $F$
  satisfies $(c_1)$ and $(c_2)$. Therefore it suffices to show that
  $F$ satisfies $(c_3)$.

  Note that we may assume without loss of generality that $F'_S$ is a
  forest. Now suppose that $F'_S$ has an edge $e$ whose endpoints are
  in different parts of $\alpha$. Since $e$ is not incident to $v$, it
  follows that $e$ is in $H$. But the partition of $V(H)$ induced by
  the connected components of $H$ is a refinement of $\alpha$, which
  is a contradiction.

  Let $u$ and $w$ be two vertices in the same connected component of
  $F$. It follows that the unique path of $F'_S$ between $u$ and $w$
  does not pass through $v$, and hence $u$ and $w$ are in the same
  part of $\alpha$ (since otherwise the path between $u$ and $w$ has
  an edge with both endpoints in different parts of $\alpha$).
  Therefore the partition of $V(H)$ induced by the connected
  components of $F$ is a refinement of $\alpha$. Conversely, let $u$
  and $w$ be two vertices contained in the same part of $\alpha$.
  Since $\alpha(v, S)$ is coarser than $\alpha$, $u$ and $w$ are in
  the same connected component of $F'_S$. Let $P$ be the unique path
  in $F'_S$ between $u$ and $w$. If $P - v$ is a path, $u$ and $w$ are
  connected in $F$. Therefore we may assume that $v$ is an internal
  vertex of $P$. Let $u', w'$ be the two neighbors of $v$ on $P$,
  where $u'$ is on the subpath of $P$ from $u$ to $v$. Since there is
  a path between $u$ and $u'$ in $F$ (namely, the subpath of $P$ from
  $u$ to $u'$), it follows from the previous argument that $u$ and
  $u'$ are in the same part of $\alpha$. Similarly, $w$ and $w'$ are
  in the same part of $\alpha$. Therefore $S$ has two vertices in the
  same part of $\alpha$, which is a contradiction.  Thus $\alpha$ is a
  refinement of the partition of $V(H)$ induced by the connected
  components of $F$. It follows that $F$ satisfies $(c_3)$ as well,
  and hence $c(i, H, \alpha)$ is at most the right hand side of the
  equation.

  Conversely, let $F$ be an optimal solution for the subproblem $(i,
  H, \alpha)$. Since $F$ is also a solution for the subproblem $(j, H,
  \alpha)$, it follows that $c(i, H, \alpha)$ is at least the right
  hand side of the equation.
\end{proofof}

~\\
\begin{proofof}{Equation~\ref{btw:eq-introduce}}
  Suppose that $v$ is not in $H$. Let $F$ be an optimal solution for
  the subproblem $(i, H, \alpha)$. Since $v$ is not in $H$, $F$ is a
  solution for the subproblem $(j, H, \alpha)$ as well, of cost
  \begin{eqnarray*}
    \len(F) + \pen(V(G_j) - V(F)) &=& \len(F) + \pen(V(G_i) -
    V(F)) - \pen(v)\\
    &=&  c(i, H, \alpha) - \pen(v)
  \end{eqnarray*}
  Thus
  $$c(i, H, \alpha) \geq c(j, H, \alpha) + \pen(v)$$
  Conversely, let $F$ be an optimal solution for the subproblem $(j,
  H, \alpha)$. Then $F$ is a solution for $(i, H, \alpha)$ of cost
  \begin{eqnarray*}
    \len(F) + \pen(V(G_i) - V(F)) &=& \len(F) + \pen(V(G_j) -
    V(F)) + \pen(v)\\
    &=& c(j, H, \alpha) + \pen(v)
  \end{eqnarray*}
  Thus
  $$c(i, H, \alpha) \leq c(j, H, \alpha) + \pen(v)$$
  Therefore we may assume that $v$ is in $H$. Let $\alpha'$ be a
  partition of $V(H) - \{v\}$ satisfying the conditions above, and let
  $F'$ be an optimal solution for the subproblem $(j, H - v,
  \alpha')$. Let $F = F' \cup E(v, S) \cup \{v\}$, where $E(v, S)$ is
  the set of all edges of $H$ that are incident to $v$. Now we claim
  that $F$ is a solution for the subproblem $(i, H, \alpha)$.  By
  construction, $F$ satisfies $(c_1)$ and $(c_2)$. Therefore it
  suffices to verify that $F$ satisfies $(c_3)$.

  Let $u$ and $w$ be two vertices in the same connected component of
  $F$. Suppose that $u$ and $w$ are connected in $F'$. Then $u$ and
  $w$ are in the same part of $\alpha'$ and, since $\alpha$ is coarser
  than $\alpha'$, $u$ and $w$ are in the same part of
  $\alpha$. Therefore we may assume that $u$ and $w$ are not connected
  in $F'$. Thus $u$ and $w$ are in different parts of $\alpha'$, each
  of which contains a vertex in $S$. It follows that the two parts
  have merged into a single part of $\alpha'(v, S) = \alpha$, and
  hence $u$ and $w$ are in the same part of $\alpha$.  Conversely, let
  $u$ and $w$ be two vertices in the same part of $\alpha$. If $u$ and
  $w$ are in the same part of $\alpha'$, it follows that $u$ and $w$
  are connected in $F'$. Therefore we may assume that $u$ and $w$ are
  in different parts $P_1$ and $P_2$ of $\alpha'$, each of which
  contains a vertex of $S$. Let $u'$ and $w'$ be the two vertices of
  $P_1 \cap S$, $P_2 \cap S$. Then there exists a path in $F'$ from
  $u$ to $u'$, and a path from $w$ to $w'$. These two paths together
  with the edges $u'v$, $vw'$ form a connected subgraph of $F$. It
  follows that $u$ and $w$ are connected in $F$, and hence $F$
  satisfies $(c_3)$.

  We have
  $$\len(F) = \len(F') + \sum_{uv \in G} \len(uv)$$
  Since $V(F) = V(F') \cup \{v\}$ and $V(G_i) = V(G_j) \cup \{v\}$,
  $$\pen(V(G_i) - V(F)) = \pen(V(G_j) - V(F'))$$
  Thus
  $$c(i, H, \alpha) \leq c(j, H - v, \alpha') + \sum_{uv \in H}
  \len(uv)$$
  Conversely, let $F$ be an optimal solution for the subproblem $(i,
  H, \alpha)$. Without loss of generality, $F$ is a forest. Let $F' =
  F - v$, and let $\alpha'$ be the partition of $V(H) - \{v\}$ induced
  by $F'$. Since $F$ is a forest, $v$ has at most one neighbor in each
  part of $\alpha'$. Now we claim that $\alpha'(v, S) = \alpha$. Let
  $T$ be any connected component of $F$ that does not contain $v$. It
  follows that $T$ is a connected component of $F'$ as well. Since the
  partition of $V(H)$ induced by the connected components of $F$ is
  equal to $\alpha$, $\alpha'$ contains each part of $\alpha$ that
  does not intersect $S$. Now consider the connected component $T$ of
  $F$ that contains $v$, and let $T_1, \cdots, T_{\ell}$ be the
  connected components of $T - v$. Since each $T_j$ is a connected
  component of $\alpha'$, it follows that the part of $\alpha'(v, S)$
  containing $v$ can be obtained by merging the parts of $\alpha'$
  induced by $T_1,\cdots, T_{\ell}$ into a single part, and adding $v$
  to it. Thus $\alpha'(v, S) = \alpha$.

  Now we claim that $F'$ is a solution for the subproblem $(j, H - v,
  \alpha')$. By construction, $F'$ satisfies $(c_1)$ and $(c_2)$.
  Additionally, it follows from the definition of $\alpha'$ that $F'$
  satisfies $(c_3)$. We have
  $$\len(F) = \len(F') + \sum_{uv \in H} \len(uv)$$
  As before,
  $$\pen(V(G_i) - V(F)) = \pen(V(G_j) - V(F'))$$
  Thus
  $$c(i, H, \alpha) \geq c(j, H - v, \alpha') + \sum_{uv \in H} \len(uv)$$
\end{proofof}

\end{document}